%% file: Tenderstake.tex
\documentclass[sigconf]{acmart}
\AtBeginDocument{%
  \providecommand\BibTeX{{%
    \normalfont B\kern-0.5em{\scshape i\kern-0.25em b}\kern-0.8em\TeX}}}

\renewcommand\footnotetextcopyrightpermission[1]{} 

\author{Conor McMenamin}
\affiliation{
\institution{Universitat Pompeu Fabra}
  \city{Barcelona}
  \country{Spain}}
\email{conor.mcmenamin@upf.edu}

\author{Vanesa Daza}
\affiliation{
  \institution{Universitat Pompeu Fabra}
  \city{Barcelona, Spain}\\
  \country{CYBERCAT - Center for Cybersecurity Research of Catalonia}
  }
\email{vanesa.daza@upf.edu}

\author{Matteo Pontecorvi}
\affiliation{
\institution{NOKIA Bell Labs}
  \city{Nozay}
  \country{France}}
\email{matteo.pontecorvi@nokia.com}

\renewcommand{\shortauthors}{McMenamin, Daza and Pontecorvi}

\usepackage{xcolor}
\usepackage{graphicx}
\usepackage{graphics}
\usepackage{eso-pic}
\usepackage{float}

\usepackage{mathrsfs}
\usepackage{algorithm}
\usepackage[noend]{algpseudocode}

\usepackage{array, makecell} 
\usepackage{tikz}
\usepackage{multirow}
\usepackage{epstopdf}

\usepackage{caption}
\usepackage{footnote}
\makesavenoteenv{tabular}

\newcommand\With{\textbf{with}}
\newcommand\WhileLocal{\textbf{while}}
\newcommand\From{\textbf{from}}
\newcommand\Broadcast{\textbf{broadcast}}

\usepackage{amsmath}
\usepackage{amsthm}
\theoremstyle{definition}

\newtheorem{theorem}{Theorem}[section]

\newtheorem{lemma}[theorem]{Lemma}

\newtheorem{definition}[theorem]{Definition}

\newtheorem{corollary}[theorem]{Corollary}

\newtheorem{observation}[theorem]{Observation}
\newtheorem{notation}[theorem]{Notation}

\newtheorem{remark}[theorem]{Remark}

\input{ACMSubmissionData} 
\input{VariableList}

\newcommand\nnfootnote[1]{%
  \begin{NoHyper}
  \renewcommand\thefootnote{}\footnote{#1}%
  \addtocounter{footnote}{-1}%
  \end{NoHyper}
}

\begin{document}

\title{Achieving State Machine Replication without Honest Players}

\input{Abstract}



\maketitle

\nnfootnote{\begin{minipage}{0.06\textwidth}
    \includegraphics[width=\linewidth]{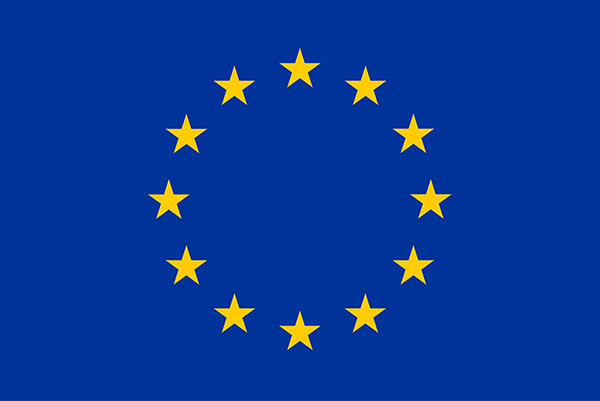}
    \end{minipage}%
    \hfill%
    \begin{minipage}{0.4\textwidth} This Technical Report is part of a project that has received funding from the European Union's Horizon 2020 research and innovation programme under grant agreement number 814284 \end{minipage}} 
\settopmatter{printfolios=true}

\input{Introduction/Introduction}
\input{RelatedWork}

\input{Prelims/Prelims}

\input{ByRaFramework}
\input{ProofsForSICIEFair}
\input{Protocol/Protocol}
\input{AchievesSMR/AchievesSMR}
\input{Conclusion}

\bibliographystyle{ACM-Reference-Format}
\bibliography{references}

\appendix

\end{document}

%% file: ACMSubmissionData.tex
\setcopyright{acmcopyright}
\copyrightyear{2021}
\acmYear{2021}
\acmDOI{10.1145/1122445.1122456}

\acmConference[Technical Report]{Technical Report}{May 31, 2021}{Barcelona}
\acmBooktitle{TR: Technical Report}
\acmPrice{15.00}
\acmISBN{978-1-4503-XXXX-X/18/06}
\acmSubmissionID{123-A56-BU3}

%% file: VariableList.tex
\usepackage{cancel}

\usepackage{pifont}
\newcommand{\yes}{\ding{51}}
\newcommand{\no}{\ding{55}}

\newcommand{\action}{x}
\newcommand{\actionSet}{X}
\newcommand{\adversary}{\mathcal{A}}
\newcommand{\adversarialBound}{\alpha}
\newcommand{\block}{B}
\newcommand{\blockchain}{\mathcal{C}}
\newcommand{\blockchainHeight}{H}

\newcommand{\deliverTime}{r_{deliver}}
\newcommand{\deviationProofs}{\textit{DevProofs}}
\newcommand{\playerDeviationProofs}{\deviationProofs_i}

\newcommand{\epoch}{\textit{epoch}}
\newcommand{\GST}{\textit{GSR}}
\newcommand{\heightVariable}{h}

\newcommand{\initialisationHeight}{1}
\newcommand{\initialisationRound}{1}
\newcommand{\initialisationEpoch}{1}

\newcommand{\notStrictlyDominated}{\text{NSD}}

\newcommand{\numByzantines}{f}

\newcommand{\numPlayers}{n}

\newcommand{\ourBAR}{\text{ByRa}}
\newcommand{\ourSMR}{\ourBAR\text{ SMR}}
\newcommand{\partialSyncAdvLimit}{\frac{1}{3}}
\newcommand{\partialSyncMajority}{\frac{2}{3}}
\newcommand{\player}{P}

\newcommand{\protocol}{\Pi}
\newcommand{\protocolName}{\text{Tenderstake}}
\newcommand{\proofProposal}{\textit{proofProposal}}

\newcommand{\removeSlashers}{\textit{adjustForSlashing}}
\newcommand{\reward}{\textit{Reward}}
\newcommand{\rewardInitial}{\textit{Genesis.reward}()}
\newcommand{\round}{r}
\newcommand{\roundprime}{{r'}}
\newcommand{\roundprimeprime}{{r''}}
\newcommand{\secParam}{\kappa}
\newcommand{\share}{s}
\newcommand{\shareSet}{\textit{Shares}}

\newcommand{\SICIE}{\text{SINCE}}
\newcommand{\slashedShare}{\textit{slashedShare}}

\newcommand{\SMRUpdate}{V}
\newcommand{\stake}{\textit{Stake}}
\newcommand{\step}{\textit{step}}
\newcommand{\strategy}{\textit{str}}
\newcommand{\strategySet}{\textit{Str}}
\newcommand{\strategySeti}{\strategySet_i}
\newcommand{\strictlyDominates}{>_u}
\newcommand{\sendTime}{r_{send}}
\newcommand{\timeout}{\tau}
\newcommand{\timeouti}{\timeout_i}
\newcommand{\type}{t}
\newcommand{\typeSet}{T}

\newcommand{\utility}{u}
\newcommand{\val}{v}

\newcommand{\actioni}{\action_i}
\newcommand{\actionSeti}{\actionSet_i}
\newcommand{\approxOne}{\approxOneValue}
\newcommand{\approxOneValue}{1- \ \negligible}
\newcommand{\blockGenesis}{\block^\initialisationHeight}
\newcommand{\blockH}{\block^\blockchainHeight}

\newcommand{\equivR}{\equiv^\roundCurrent}
\newcommand{\expectedUtilityiRound}{\overline{\utility}^\roundprime_i}

\newcommand{\negligible}{\textit{negl}(\kappa)}
\newcommand{\notequiv}{\cancel{\equiv}}

\newcommand{\playerIndices}{\{ 1,...,\numPlayers\}}
\newcommand{\playeri}{\player_i}
\newcommand{\playerReward}{\reward_i}
\newcommand{\playerSetDescription}{\{\player_1,...,\player_\numPlayers\}}
\newcommand{\playerShares}{\shareSet_i}
\newcommand{\playerStake}{\stake_i}
\newcommand{\poly}{\text{polynomial in }\secParam}

\newcommand{\roundCurrent}{\round}

\newcommand{\shareH}{\share^\blockchainHeight}
\newcommand{\shareHi}{\shareH_i}
\newcommand{\sharei}{\share_i}

\newcommand{\stakeH}{\stake^\blockchainHeight}

\newcommand{\strategyi}{\strategy_i}
\newcommand{\strategyProtocol}{\strategy_\Pi}
\newcommand{\strategySetNonDominatedi}{\strategySet^\notStrictlyDominated_i}
\newcommand{\typei}{\type_i}
\newcommand{\typeir}{\typei^\roundCurrent}
\newcommand{\typeSeti}{\typeSet_i}
\newcommand{\typeSetir}{\typeSeti^\roundCurrent}
\newcommand{\utilityi}{\utility_i}

\newcommand{\utilityir}{\utilityi^\round}

%% file: Abstract.tex
\begin{abstract}
  Existing standards for player characterisation in tokenised state machine replication protocols depend on honest players who will always follow the protocol, regardless of possible token increases for deviating. Given the ever-increasing market capitalisation of these tokenised protocols, honesty is becoming more expensive and more unrealistic. As such, this out-dated player characterisation must be removed to provide true guarantees of safety and liveness in a major stride towards universal trust in state machine replication protocols and a new scale of adoption. As all current state machine replication protocols are built on these legacy standards, it is imperative that a new player model is identified and utilised to reflect the true nature of players in tokenised protocols, now and into the future.
  
  To this effect, we propose the $\ourBAR$ player model for state machine replication protocols.  In the $\ourBAR$ model, players either attempt to maximise their tokenised rewards, or behave adversarially. This merges the  fields of game theory and distributed systems, an intersection in which tokenised state machine replication protocols exist, but on which little formalisation has been carried out. In the $\ourBAR$ model, we identify the properties of strong incentive compatibility in expectation and fairness that all protocols must satisfy in order to achieve state machine replication. We then provide $\protocolName$, a protocol which provably satisfies these properties, and by doing so, achieves state machine replication in the $\ourBAR$ model.

\end{abstract}

%% file: Introduction/Introduction.tex
\section{Introduction}\label{sec:intro}

Current state machine replication (SMR) protocols, a subset of which being blockchain protocols, depend on the existence of altruistic players who ignore token changes and honestly follow the protocol. If a player can deviate from a protocol to increase their tokens with no perceived effect on safety and liveness, it must be assumed that every such individual will choose to do this. In Flash Boys 2.0 \cite{FlashBoys2.0} and subsequent work\footnote{\url{https://github.com/flashbots/pm} Accessed: 25/05/2021}, it is demonstrated that these deviation opportunities are rampant in Ethereum, and that players are actively availing of them. In any large-scale SMR protocol, most, if not all players, will not consider their deviations as affecting SMR. Therefore, it is essential that we assume non-adversarial players will seek to maximise tokens in tokenised protocols. As a direct consequence, SMR guarantees can no longer depend on honest-by-default users. We explicitly outline the $\ourBAR$ (Byzantine or Rational) model as an updated player characterisation framework to reflect this weakness in current standards. By moving to the $\ourBAR$ model, which we formally define in Definition \ref{def:PlayerModel}, the caveat of honest player dependencies in current SMR protocols is removed. Furthermore, we demonstrate that it is possible to achieve SMR in the $\ourBAR$ model by providing the $\protocolName$ protocol, an amendment to the Tendermint protocol \cite{Tendermint,LatestGossipTendermint}. 

To progress towards global adoption, a tokenised SMR protocol must first ensure that all players will maximise their tokens by following the protocol. Implementing an SMR protocol that increases a player's tokens for following the protocol is known as incentivisation, and is a fundamental requirement for any SMR protocol. Much of the work on incentivisation in SMR protocols stems from the seminal work on \textit{selfish mining} in Nakamoto-consensus \cite{SelfishMining}. In \cite{SelfishMining}, it is demonstrated that certain players are incentivised to deviate from the prescribed protocol. This eventually leads to a scenario where SMR properties are violated, as discussed in \cite{SelfishMining}. It is only upon the performing of actions as required by the protocol by some majority that it is possible to guarantee the SMR properties of safety and liveness. This has remained the case in the age of tokenisation.

Despite this, there has been no thorough treatment and analysis of tokenised SMR protocols from a game-theoretic standpoint involving rational players, who want to maximise their net tokenised gains (referred to as utility increases in game-theoretic literature), and an adversary, who can corrupt the owners of some amount of the tokenised consensus resource and behave arbitrarily. These corrupted players are known as Byzantine. This characterisation of players as either Byzantine or Rational, which we refer to as the $\ourBAR$ model, was first considered in distributed systems literature in \cite{SelfishMeetsEvil}, but never successfully with respect to SMR protocols, although attempts have been made \cite{Fairledger,RationalsvsByzantinesConsensusBasedBlockchains,BlockchainsCannotRelyonHonesty}. The closest semblance to this model which has seen wide-scale adoption with respect to SMRs is the BAR (Byzantine, Altruistic and Rational) model \cite{BAR-FT}. The BAR model crucially includes some portion of altruistic players who disregard tokenised utility, and always follow the protocol. Examples of authors echoing our desire to move away from altruistic dependencies are numerous, but this from  Fairledger \cite{Fairledger} puts it concisely: ``We have to take into account that every entity may behave rationally, and deviate from the protocol if doing so increases its benefit". Non-adversarial, honest-by-default characters do not exist in competitive games, and cannot be depended on in tokenised SMR protocols due to their gamified nature. Although many other works state the need to move away from altruistic dependencies, none have proven the critical nature of this dependency, or provided protocols which achieve SMR, in the $\ourBAR$ model. In this paper, we fulfil both of these essential tasks.

Without the safety net of altruistic players, any successful instantiation of an SMR protocol in the $\ourBAR$ model must guarantee that rational players will always follow the protocol. To ensure this, rational players must expect to strictly maximise their utility by following the protocol, a property we define as \textit{strong incentive compatible in expectation} ($\SICIE$). 

Moreover, we must also guarantee that within such an incentive compatible protocol, the adversary cannot increase their share of tokens to a point where they control enough tokens to prevent SMR. Despite the existence of strong incentive compatibility in expectation, it may be possible for an adversary to receive more than their share of the tokens that get distributed, increasing their share of control.
Therefore, we must additionally ensure that an adversary cannot increase the share of tokens they control, a property we define as \textit{fairness}.

\input{Introduction/OurContribution}

\input{Introduction/PaperOrganisation}

%% file: Introduction/OurContribution.tex
\subsection{Our Contribution}

We define the $\ourBAR$ player characterisation model, the properties of $\SICIE$ and fairness, and in Definition \ref{def:SMRFails}, the basic requirements a prospective SMR protocol must meet in order to guarantee safety and liveness in the $\ourBAR$ model. If these requirements are met for a protocol in the $\ourBAR$ model, the protocol \textit{achieves $\ourSMR$}. Informally, to achieve $\ourSMR$ we require that players controlling a majority of tokens follow the protocol at all times. 
We then prove that the properties of $\SICIE$ and fairness are necessary and together sufficient to achieve $\ourSMR$ in the main theorem of the paper.

\begingroup
\def\thetheorem{\ref{thm:fair+SICIE=Necessary+Sufficient}}
\begin{theorem}
    For an SMR protocol $\protocol$, $\protocol$ achieves $\ourSMR$ if and only if $\protocol$ is strong incentive compatible in expectation and fair. 
\end{theorem}
\addtocounter{theorem}{-1}
\endgroup

In addition to this new game-theoretical framework, we provide $\protocolName$ as a concrete instantiation of an SMR protocol that provably achieves $\SICIE$ and fairness in the $\ourBAR$ model. Using Theorem \ref{thm:fair+SICIE=Necessary+Sufficient}, we then prove $\protocolName$ achieves SMR in the $\ourBAR$ model.

%% file: Introduction/PaperOrganisation.tex
\subsection{Organisation of the paper}

In Section \ref{sec:relatedWork} we review related work and present an overview of attempts to implement, and works in favour of, the $\ourBAR$ model for SMR protocols. In Section \ref{sec:SystemModel} we provide a background on the SMR and game theory concepts needed to define the $\ourBAR$ model. Section \ref{sec:ourSMRDefinitions} introduces a new game-theoretic framework for analysing SMR protocols. This new framework defines the $\ourBAR$ model, and outlines what we require from SMR protocols in the $\ourBAR$ model, introducing the properties of $\SICIE$ and fairness.
In Section \ref{sec:Properties} we prove that $\SICIE$ and fairness are necessary for a protocol to achieve $\ourSMR$. We then prove that together, $\SICIE$ and fairness are sufficient properties for a protocol to achieve $\ourSMR$.
In Section \ref{sec:Protocol} we outline the $\protocolName$ protocol as an example, for the first time in literature, of a $\SICIE$ and fair $\ourSMR$ protocol. 
In Section \ref{sec:ProtocolProofs} we reason that $\protocolName$ satisfies the necessary and sufficient properties of safety and liveness for SMR when players controlling a majority of the consensus votes follow the protocol in every round. 
We then prove that the $\protocolName$ protocol is $\SICIE$ and fair, which using Theorem \ref{thm:fair+SICIE=Necessary+Sufficient}, implies $\protocolName$ achieves $\ourSMR$. We conclude in Section \ref{sec:conclusion}.

%% file: RelatedWork.tex
\section{Related Work}\label{sec:relatedWork}

There is a growing appreciation  that incentivisation is not only important, but necessary, to ensure the successful instantiation of an SMR protocol. Many works have argued for the incentivisation of players in SMR protocols  \cite{Ouroboros, BlockchainWithoutWastePoS,EvolutionOfSharesPoS,RationalBehaviorCommitteeBasedBlockchains, RationalsvsByzantinesConsensusBasedBlockchains, BlockchainsCannotRelyonHonesty, CasperIncentives, FruitChains,SnowWhite,SoKConsensus,SOKGameTheoryBlockchain,SmartCast, SelfishMeetsEvil, SoKToolsGameTheoryCryptocurrencies, Solida} while many others demonstrate the critical need for incentive compatibility in tokenised SMR protocols \cite{FlashBoys2.0,FairnessTendermint2019,BitcoinANaturalOligopoly, BitcoinRnDArmsRace, BlockchainFolkTheorem, SelfishMining, PoSCompounding, Bitcoin, IntermittentSelfishMinig, FairnessDAGs, EconomicLimitsBitcoinBlockchain, EIP1559roughgarden}. 


\begin{table*}
    
  \begin{tabular}{lccccc}
    \toprule
    Paper & Network Model & \makecell{Player Model \\ w/o Honest Players} & Evolving-Stake Adversary & SINCE & Fair \\
    \midrule
    Rationals vs. Byzantines \cite{RationalsvsByzantinesConsensusBasedBlockchains} & Broadcast Synchrony\footnotemark & \textcolor{green}{\yes} & \textcolor{red}{\no} & \textcolor{green}{\yes} \footnotemark & \textcolor{green}{\yes} \footnotemark   \\  
    Blockchain Without Waste \cite{BlockchainWithoutWastePoS} & Synchrony & \textcolor{green}{\yes} & \textcolor{red}{\no}\footnotemark & \textcolor{green}{\yes} \footnotemark  & \textcolor{red}{\no}  \\  
    Blockchains Cannot Rely on Honesty \cite{BlockchainsCannotRelyonHonesty} & Synchrony & \textcolor{green}{\yes}  & \textcolor{red}{\no} & \textcolor{red}{\no} & \textcolor{red}{\no}  \\  
    Fruitchains, Snow White \cite{FruitChains,SnowWhite} & Partial Synchrony & \textcolor{red}{\no} & \textcolor{red}{\no} & \textcolor{red}{\no} & \textcolor{red}{\no}  \\  
    Casper Incentives \cite{CasperIncentives} & Partial Synchrony & \textcolor{red}{\no}  & \textcolor{red}{\no} & \textcolor{red}{\no} & \textcolor{red}{\no}  \\ 
    FairLedger \cite{Fairledger} & Synchrony & \textcolor{red}{\no}  & \textcolor{red}{\no} & \textcolor{red}{\no} & \textcolor{red}{\no}  \\  
    \textbf{$\protocolName$} (Algorithm \ref{alg:tenderstake}) & Partial Synchrony & \textcolor{green}{\yes} & \textcolor{green}{\yes} & \textcolor{green}{\yes} & \textcolor{green}{\yes}  \\
    \bottomrule
  \end{tabular}
  \caption{Comparison of main works claiming incentive compatibility\label{table:RelatedWorkComparison}.
  \footnotesize{\textsuperscript{2}An idealised model where every message, including adversarial messages, are known to be instantly delivered to all players.
  \textsuperscript{3}No explicit reward mechanism provided, non-trivial for BFT protocols.
  \textsuperscript{4}Enforced by the idealised network model/ unspecified reward mechanism.
  \textsuperscript{5}No adversary in player model.
  \textsuperscript{6}Author creates a dominating cost unrelated to quantity of stake for deviation.} }
  
\end{table*}

The characterisations of Byzantine and rational, coupled with that of altruistic players who always follow the protocol, segues into the BAR player characterisation model as introduced in \cite{BAR-FT}. However, as discussed in Section \ref{sec:intro}, tokenised SMR protocols cannot depend on altruistic players to ensure the critical properties of safety and liveness. 
We amend the player characterisations to only include those of Byzantine and rational players in what we call the $\ourBAR$ model. 

A very similar player model is discussed in \cite{SelfishMeetsEvil}, but with respect to a single binary action multiparty computation. We extend this basic binary action space for players to allow for indefinite sequentialised non-binary action profiles in line with those of SMR protocols. We introduce the necessity for strict maximisation of expected utility to ensure rational players always follow a protocol. This is opposed to \cite{SelfishMeetsEvil}, where it is claimed that equality of utility will suffice to ensure a rational player will choose one strategy over another. This is logically insufficient. Related to this concept of insufficient proof mechanisms, a common pitfall of legacy incentive compatible proofs is to prove that following a protocol is a Nash Equilibrium in the presence of honest players \cite{FruitChains,SnowWhite,Ouroboros,PoSCompounding}. In the $\ourBAR$ model this assumption is not possible, and therefore those proofs are not sound.  We also allow the adversary to behave arbitrarily, as opposed to \cite{SelfishMeetsEvil} where the adversary only tries to minimise the utility of rational players. Although there are buzzwords associated with this paper such as \textit{Price of Malice} and \textit{Price of Anarchy}, no name is attributed to the player model. We refer to our version of this player model as the $\ourBAR$ model. The only examples of this player model in SMR literature making meaningful attempts to remove altruistic entities are in \cite{RationalsvsByzantinesConsensusBasedBlockchains,BlockchainsCannotRelyonHonesty}.

Table \ref{table:RelatedWorkComparison} exhibits the shortcomings of related work in providing protocols that guarantee rational players always follow the protocol (SINCE), and that prevent an adversary from increasing their share of stake to destroy the system (Fair). Table \ref{table:RelatedWorkComparison} also includes our proposal, $\protocolName$ as a standard against which to compare these works.

In \cite{RationalsvsByzantinesConsensusBasedBlockchains}, it is implicitly assumed rewards are paid to all players who contribute to consensus on a block. This is non-trivial in the $\ourBAR$ model, as rewards in their system depend on message delivery. From a protocol's perspective, these messages need to recorded by a proposer at some point in the protocol, and rational proposers may be incentivised to omit players, as is the case in previous works from subsets of the same authors \cite{FairnessTendermint2018,FairnessTendermint2019}. We address this omission in the $\protocolName$ protocol, providing an explicit solution in the $\ourBAR$ model. 

Although \cite{BlockchainsCannotRelyonHonesty} provides an SMR protocol which approaches SINCE, they do not provide a rigorous player model excluding altruistic players, and in the presence of a deviating adversary, there are strategies which strictly outperform the recommended protocol strategy for rational players, preventing both strong incentive compatibility and fairness. 

A purely economic approach to SMR protocols is taken in \cite{BlockchainWithoutWastePoS}, which focuses on Proof-of-Stake protocols. Their player model only considers rational players, and depends on a dominating cost for certain deviations that is not quantifiable within the protocol game of maximising stake. Namely, the author assumes rational players in a longest chain rule Proof-of-Stake system will never try to fork the blockchain, as doing so devalues stake in terms of some external fiat currency more than any possible reward. We believe this does not necessarily affect the decisions of all rational players, which is also acknowledged in \cite{BlockchainWithoutWastePoS} where participation in the protocol is restricted to players with a ``sufficient coin holding''. Another concern about such an arbitrary external cost arises when we consider settings where the stake/ cryptocurrency in question becomes a dominant fiat currency, and the majority of participants only consider utility as measured in said stake. In this paper, we demonstrate that it is possible to construct a protocol, $\protocolName$, that strictly maximises stake by following the protocol. As following the protocol maximises the value of stake in \cite{BlockchainWithoutWastePoS}, $\protocolName$ captures the same maximisation of value without the potentially problematic dependency on unquantifiable external costs unrelated to quantity of stake.

One of the legacy works in relation to fairness and incentive compatibility of SMR protocols is Fruitchains \cite{FruitChains}. The Fruitchains player model consists of an altruistic majority of players and a cooperative rational minority. Fruitchains crucially relies on an underlying blockchain satisfying an SMR protocol in order to guarantee fairness of rewards. They fail to consider the incentives of all parts of the system, relying on an altruistic majority in order to guarantee the underlying blockchain satisfies the required SMR properties. They then add a small section where claims of incentive compatibility for non-cooperative rational players are made. The authors claim a protocol is incentive compatible if fairness of rewards has already been guaranteed. As fairness in their system is only guaranteed if a majority of players follow the protocol, there is no logical result which proves that rational players will always follow the protocol, required for incentive compatibility. This is insufficient to guarantee SMR in the $\ourBAR$ model. This fatal dependence on an underlying correct-by-default SMR protocol/ trusted third-party is also demonstrated in \cite{CasperIncentives,Fairledger}, where claims of incentive compatibility and fairness do not hold in the $\ourBAR$ model.

%% file: Prelims/Prelims.tex
\section{Preliminaries}\label{sec:SystemModel}

This section covers the concepts and definitions required to reason about SMR protocols from a game-theoretic perspective. First we define SMR and  a general notion of a blockchain which provides some intuition for our SMR definitions, and primes the reader for our description of the $\protocolName$ protocol in Section \ref{sec:Protocol}. We then provide the game theory framework necessary to formally reason about SMR protocols involving rational and adversarial players, and how SMR can be achieved in the presence of these types of players. In the following we let $\textit{negl}()$ be a function which for any polynomial $\textit{p}()$ there exists a constant $\kappa_0 \in \mathbb{N}$ such that $\textit{negl}(\kappa) < \frac{1}{\textit{p}(\kappa_0)}$ for all $\kappa \geq \kappa_0$. This $\textit{negl}()$ is known in literature as a \textit{negligible} function. 

\input{Prelims/SMRBackground}

\input{Prelims/GameTheoryBackground}

%% file: Prelims/SMRBackground.tex
In this paper, we are interested in a distributed set of $\numPlayers$ players $\playerSetDescription$ interacting with one and other inside a protocol which will produce some output that all players correctly participating in the protocol can agree on. This output will be a \textit{replicated state machine}. First, we define a state machine.

\begin{definition}
    A \textit{state machine} consists of set of variables, and sequence of commands/ updates on those variables, producing some output. 
\end{definition}

The concept of a state machine alone does not capture the notion that potentially many players can reconstruct a common view of the same state of a machine, and requires extension.

\begin{definition}
    For a set of players $\playerSetDescription$ and a state machine, \textit{state machine replication} (SMR) is a process that allows each player to execute a common sequence of commands acting on the machine's state in the same order, thus maintaining a common view of the machine's state.
\end{definition}

Progressing towards our goal of analysing SMR protocols, we must first define what we require from an SMR protocol. We take inspiration for our definition from \cite{SyncByzCon}, where their system model is clearly and concisely explained, and is very similar to ours. 

\begin{notation}
    With respect to protocols and recommended protocol actions, a \textit{correct} player is a player who always follows the recommended protocol actions.
\end{notation}

\begin{definition}\label{def:SMR}
An \textit{SMR protocol} $\protocol$ deciding on a potentially infinite sequence of state machine updates satisfies the following properties:
\begin{itemize}
    \item \textit{Safety}: For any two correct players $\playeri$, $\player_j$ in $\protocol$, $i\neq j$, if $\playeri$ decides on an SMR update $\SMRUpdate_i$ at position $k$ in the sequence, and $\player_j$ decides on an SMR update $\SMRUpdate_j$ at position $k$ in the sequence, then $\SMRUpdate_i = \SMRUpdate_j$.
    \item \textit{Liveness}: For any position $k$ in the sequence, every correct player eventually decides on an SMR update for position $k$.
\end{itemize}
\end{definition}

To achieve SMR, we utilise the concept of a blockchain. This is done in a generic manner so as to allow for direct comparison with most blockchain instantiations.

\begin{definition}
    A \textit{block} $\block$ is a data structure used to communicate changes to the state machine view of each player. Blocks consist of a pointer(s) to previous block(s), and a set of instructions with which to update the state. State machine updates in a block are applied to the state described by the block(s) to which they point. The \textit{genesis block} $\blockGenesis$ describes the starting state of the system and is a priori agreed upon by all players. The global state at any point in the system is then described by applying the state machine updates according to some ordering rule starting from the genesis block. A \textit{blockchain} $\blockchain= [ \blockGenesis,...,\blockH ]$ is the ordered data structure created by traversing the block pointers from the genesis block to all blocks to be applied to the global state according to the ordering rule. $\blockchainHeight$ denotes the \textit{height} of the blockchain.
\end{definition}

In our system, an SMR protocol $\protocol$ consists of $\numPlayers$ players owning shares of a finite resource, which we will refer to as \textit{stake}, and denoted $\stake^\initialisationRound$ at initialisation. $\protocol$  proceeds in fixed-time periods, which we refer to as \textit{rounds}, beginning in round $\initialisationRound$.
For any height $\blockchainHeight > \initialisationHeight$ of the blockchain, players participate in $\protocol$ to decide on a block for that height. Reaching consensus on a block will involve one or more successful protocol steps. After a block has been decided for height $\blockchainHeight\geq \initialisationHeight$, the total stake in the system is denoted $\stakeH$ with player shares of $\stakeH$ denoted $\shareH_1,....,\shareH_\numPlayers$. Without loss of generality, we assume $\sum^{\numPlayers}_{i=1}\shareH_i=1$, and for all $  i \in \playerIndices, \, \blockchainHeight \geq1$, $\shareH_i<\frac{1}{2}$.

%% file: Prelims/GameTheoryBackground.tex
Now we introduce some basic game theory to allow us to properly reason about SMR protocols in our system as games, taking inspiration for our definitions from \cite{AlgoGameTheory}. The games we are concerned with, SMR protocols, are played by players with strict incomplete information, meaning some subset of players will not know the action choices of other players for certain rounds when they are required to choose their own actions. As such, we need to be able to describe what a player knows (and implicitly what they do not), which we call their private information. Furthermore, we must be able to describe what motivates players in games. This motivation is provided by a utility function, which attributes a numerical score to each action a player can take. In games, players choose the action which maximises their utility function.

\begin{definition}\label{def:game}
A \textit{game}, denoted $\mathbb{G}$, progressing in rounds with strict incomplete information for a set of $\numPlayers$ players $\playerSetDescription$ can be described by the following:
 \begin{itemize}
    \item For every $\playeri$, a set of \textit{actions} $\actionSeti$. We denote by $\actionSet_{-i}$ the set of actions that each player excluding $\playeri$ can take. 
    For $\action_{-i} \in \actionSet_{-i}$, $\action_{-i}$ is described by a vector of actions of length $\numPlayers-1$, with each vector position mapping to a unique player.
    \item For every player $\playeri$ and round $\roundCurrent$, a set of \textit{private informations} $\typeSetir$. 
    A value $\typeir \in \typeSetir$ is a private information value that $\playeri$ can have at round $\roundCurrent$. We denote by $\type_{-i}^\roundCurrent$ the private informations held by all players excluding $\playeri$ at round $\roundCurrent$.
    
    \item For every player $\playeri$, current round $\roundCurrent \geq \initialisationRound$, and some round $\roundprime \geq \roundCurrent$, the \textit{utility function } for $\playeri$ with respect to round $\roundprime$ is defined as : 
    \begin{equation}
        \utilityir: \typeSeti^\roundCurrent \times \underbrace{\actionSeti \times ... \times \actionSeti}_{\roundprime+1 - \roundCurrent} \times \underbrace{\actionSet_{-i} \times ... \times \actionSet_{-i}}_{\roundprime+1 - \roundCurrent} \to \mathbb{R}
    \end{equation}
    where $\utilityir(\typeir, \actioni^\roundCurrent,..., \actioni^\roundprime,\action_{-i}^\roundCurrent,..., \action_{-i}^\roundprime )$ is the utility achieved by $\playeri$ in round $\roundprime$ with private information $\typeir$, if player $\playeri$ takes the actions $\actioni^\roundCurrent,..., \actioni^\roundprime$ in rounds $\roundCurrent,...,\roundprime$ respectively,  and the actions of all other players are described by $\action_{-i}^\roundCurrent,..., \action_{-i}^\roundprime$ in rounds $\roundCurrent,..., \roundprime$ respectively.
 \end{itemize}  
\end{definition}

Although utility functions evaluate actions given the actions of all other players, the actions of the other players may not be known in advance. Therefore, players will need to be able to choose their actions solely based on their private informations. The actions a player takes given some private information are computed through a strategy, which is defined in Definition \ref{def:strategy}.

\begin{definition}\label{def:strategy}
   A \textit{strategy} of a player $\playeri$ is a function $\strategyi:  \typeSetir  \to \actionSeti$, $\roundCurrent \geq \initialisationRound$, which defines the action to be taken by $\playeri$ given some private information value. A strategy $\strategyi$ is \textit{mixed} if for a player $\playeri$ with $m_i$ possible strategies $\strategySeti=\{\strategyi^1,\, ..., \strategyi^{m_i}\}$, they select a strategy to follow from $\strategySeti$ according to some probability distribution. For every player $\playeri$, $\strategy_{-i}$ describes the mixed strategies taken by all players excluding $\playeri$.
\end{definition}

\begin{definition}
   For an SMR protocol $\protocol$, the \textit{recommended strategy}, denoted $\strategyProtocol$, is the strategy that $\protocol$ requires players to follow in order to successfully achieve SMR.
\end{definition}

%% file: ByRaFramework.tex
\section{A Game-Theoretic Framework for SMR}\label{sec:ourSMRDefinitions}

In this section we formalise the $\ourBAR$ framework for SMR protocols, where participants are either adversarially or rationally motivated. This is in response to the existential threat posed by the growing trend of players managing SMR protocols acting in a profit-maximising manner \cite{FlashBoys2.0} in protocols where security guarantees depend on honest players. Furthermore, this framing is made quite naturally, given SMR protocols are accurately modelled as games with strict incomplete information as defined in Definition \ref{def:game}. 

This is a crucial progression from existing standards in distributed systems literature where some number of non-adversarial players are honest-by-default. Due to the distributed nature of SMR protocols, as a baseline we must account for some portion of adversarial players who can behave arbitrarily with unknown utility functions. With SMR protocols considered as games, the remaining non-adversarial players must follow some known utility function, and attempt to choose the actions which maximise it. To ensure the honest behaviour of rational players in this setting, following the protocol strategy must maximise the utility of rational players. We define these player characterisations here formally as the $\ourBAR$ model.

\begin{definition}\label{def:PlayerModel}
    The \textit{$\ourBAR$ model} consists of \textit{Byzantine} and \textit{Rational} players. A player is:
    \begin{itemize}
        \item \textit{Byzantine} if they deviate arbitrarily from the recommended strategy within a game with unknown utility function. Byzantine players are chosen and controlled by an adversary $\adversary$.
        \item \textit{Rational} if they choose uniformly at random from all mixed strategies which maximise their known utility function assuming all other players are rational.
    \end{itemize}
        
\end{definition}

\begin{remark}
    Our definition of rational players omits tie-breaking assumptions that bias a rational player to certain strategies over others with equal utility. For example, if we have a fair coin tossing game that costs 1 token to play and correct guesses gain 3 tokens, a rational player in our system will choose heads with probability 0.5. If we have a protocol that requires rational players to always choose heads, it is necessary to make the payoff for heads strictly greater than that of tails. 
\end{remark}

A rational player who assumes all other players are rational is known as an \textit{oblivious} rational player \cite{SelfishMeetsEvil,PriceOfMalice}. A rational player who is not oblivious knows there are players in the system controlling a non-negligible share of stake, controlled by an adversary, who may try to break safety and liveness. Adding this to the private information of a rational player adds a probability of safety and liveness failing if protocol actions are not followed by the remaining players, which becomes 1 in the presence of a maximal adversary. This outcome has a critical cost for rational players (as used in \cite{RationalsvsByzantinesConsensusBasedBlockchains,BlockchainWithoutWastePoS, SelfishMeetsEvil,PriceOfMalice}), which can be made arbitrarily high to prevent rational players from deviating from protocol actions. Under the non-oblivious assumption, all rational players will follow the protocol, and proofs of following protocol actions become trivial. 

We believe this is highly unrepresentative of rational players in SMR protocols today, particularly in light of the clear recent evidence that miners can and are deviating from protocol actions to increase their on-chain rewards \cite{FlashBoys2.0}. As such, in the rest of this paper, we assume all rational players are oblivious, and prove our main lemmas and theorems given this weakest possible assumption about adversarial share distributions.

To consider rational players in any game, it is necessary to explicitly define what their utility functions are. Inkeeping with the tokenised assumptions of our model, we let rational player utility be measured in stake as described by the blockchain. By their nature, tokenised SMR protocols require it to be expensive to deviate from the protocol actions, encouraging honest behaviour through stake rewards, and/or stake punishments for dishonest behaviour. Given the unprecedented levels of SMR protocol usage as a result of tokenisation, we see stake as the driving utility measure for the players who participate in these protocols.

As total stake is only meaningful with respect to a particular time-point, and SMR protocols are played indefinitely, rational players will seek to maximise their total stake at all possible rounds sufficiently far into the future. 
Therefore, when discussing incentivisation and player utility, it is necessary to refer to stake/share/total stake with respect to rounds. As we are using the round variable as a counter, and some rounds may be unsuccessful, it cannot be independently used to determine the height, and vice versa. Rather than add notation to relate the two, we treat them separately, and make it clear from context which is being used. When referring to stake/share/total stake with respect to particular rounds, we use superscripts involving $\round$, whereas when discussing these variables with respect to the height of the blockchain, we use superscripts involving $\blockchainHeight$.

In the $\ourBAR$ model, and SMR protocols in general, it is necessary to specify an upperbound on adversarial share of stake, below which SMR can be achieved if all non-adversarial players follow the protocol, and above which SMR cannot be guaranteed.

\begin{notation}
   For an SMR protocol $\protocol$, we denote by $\adversarialBound$ the maximal share of stake such that for players controlling greater than $1-\adversarialBound$ of the stake following the SMR protocol, safety and liveness are achieved. The exact value of $\adversarialBound$ will depend on the network distribution assumptions, in line with the results of \cite{PartialSynchronyDwork}, which must be contained in the threat model.
\end{notation}

For some security parameter $\secParam \in \mathbb{N}$, our goal is to guarantee that SMR can be achieved (that is, both  safety and liveness are satisfied) in the $\ourBAR$ model with probability greater than $\approxOne$ over any $\poly$ rounds.

We first need to introduce an equivalence relation for mixed strategies over finite rounds. When we state the protocol strategy which needs to be followed to achieve SMR, although there is an infinite number of strategy encodings, we only require players to follow strategies which result in actions as outlined by the protocol. We are indifferent to how this is achieved. If a strategy is encoded differently to the recommended protocol strategy, but  results in actions as prescribed by the protocol with probability greater than $\approxOne$ over any $\poly$ rounds, we see this as equivalent to the recommended protocol strategy.

\begin{definition}
    For a player $\playeri$ at initialisation, and round $\roundprime \geq \initialisationRound$, two mixed strategies $\strategyi^a$ and  $\strategyi^b$ are \textit{equivalent with respect to round $\roundprime$} if for all rounds $\round$, $\initialisationRound \leq \round \leq \roundprime $, and private informations $\typei^{\round} \in \typeSeti^{\round}$, it is the case that $\strategyi^a(\typei^{\round})=\strategyi^b(\typei^{\round})$. 
    We use $\strategyi^a \equiv^{\roundprime} \strategyi^b$ to denote this equivalence relation. If $\strategyi^a \equiv^{\roundprime} \strategyi^b$ for all rounds $  \roundprime \, \poly$,  $\strategyi^a $ and $\strategyi^b$ are \textit{equivalent}, denoted  by $\strategyi^a \equiv \strategyi^b$.
\end{definition}

With this equivalence relation, we can now define what it means for a protocol to achieve SMR in the $\ourBAR$ model. In this paper, after deciding on a block at height $\blockchainHeight \geq \initialisationHeight$, we denote the adversarial share of stake by $\shareH_\adversary$.

\begin{definition}\label{def:SMRFails}
   For an SMR protocol $\protocol$ and round $ \round$, let $p^\round_{\protocol}$ be the probability that players controlling more than $1-\adversarialBound$ of the total stake follow a mixed strategy $\strategy \equivR \strategyProtocol$ up to and including round $ \round$ for any $\share_\adversary^\initialisationRound<\adversarialBound$.  $\protocol$ \textit{achieves $\ourSMR$} if for all rounds $\roundprime \, \poly$ it holds that $p^\roundprime_{\protocol} $ is greater than $ \approxOne$. Otherwise, $\protocol$ \textit{fails in the $\ourBAR$ model}. 
\end{definition}

Towards the goal of achieving $\ourSMR$, we need to formally define rational utility as measured in stake. For a rational player $\playeri$ with private information $\typeir$ and round $\roundprime \geq \round$, we have:
\begin{equation}
    \utilityi^{\roundprime}(\typeir, \actioni^\round,..., \actioni^\roundprime,\action_{-i}^\round,..., \action_{-i}^\roundprime )= \sharei^\roundprime \cdot \stake^\roundprime.
\end{equation}
However, in a game with strict incomplete information as is the case in an SMR protocol, a rational player $\playeri$ with private information $\typeir $ will not know their own future private information values (required to choose their actions), the private informations of the other players, or $ \strategy_{-i}$, before choosing $\strategyi$. 
Therefore, $\playeri$ must choose the mixed strategy which maximises $\playeri$'s expected stake at round $\roundprime$, denoted $E(\sharei^\roundprime \cdot \stake^\roundprime)$, according to the probability distribution that $\playeri$ attributes to possible values for these unknowns. This distribution will be contained in $\typeir$.

Thus, knowing $\typeir$ is sufficient to calculate $\playeri$'s expected utility of a particular strategy at round $\roundprime$, which we express mathematically by $E(\sharei^{\roundprime} \cdot  \stake^{\roundprime} | \typeir, \strategyi)$.  We state this formally in Definition \ref{def:rationalUtility}.

\begin{definition}\label{def:rationalUtility}
   For an SMR protocol $\protocol$ and rational player $\playeri$ with private information $\typeir$, mixed strategy $\strategyi$, and a particular round $\roundprime \geq \round$, the \textit{expected utility of $\strategyi$ for $\playeri$ at round $\roundprime$} is denoted $\expectedUtilityiRound(\typeir, \strategyi)$ and is described by $\expectedUtilityiRound(\typeir, \strategyi)=E(\sharei^{\roundprime} \cdot \stake^{\roundprime} | \typeir, \strategyi)$. 
\end{definition}

As such, for a rational $\playeri$ in an SMR protocol $\protocol$ with private information $\typeir$, $\playeri$  will choose the mixed strategy $\strategyi$ which maximises $\expectedUtilityiRound(\typeir, \strategyi)$. To establish the existence, or not, of such a mixed strategy, we introduce an inequality in Definition \ref{def:strictDominate} which allows us to pairwise rank mixed strategies by expected utility.

\begin{definition}\label{def:strictDominate}
   For an SMR protocol $\protocol$, rational player $\playeri$ and two mixed strategies $\strategyi^a, \ \strategyi^b$,  $\strategyi^a$ \textit{strictly dominates} $\strategyi^b$ \textit{in expectation} if there exists $ \,  \roundprimeprime \geq \round $, $ \,  \roundprimeprime \, \poly$ , such that for all $ \, \roundprime > \roundprimeprime$, $\expectedUtilityiRound(\typeir, \strategyi^a)> \expectedUtilityiRound(\typeir, \strategyi^b)$. If $\strategyi^a$ strictly dominates $\strategyi^b$ in expectation, we denote this relationship by $\strategyi^a \strictlyDominates \strategyi^b$.
\end{definition}

Using the strict dominance in expectancy relationship, we can formally define what we require from an SMR protocol in order for rational players to follow the recommended protocol strategy. This requirement is strong incentive compatibility in expectation, and is defined in Definition \ref{def:SICIE}.

\begin{definition}\label{def:SICIE}
   An SMR protocol $\protocol$ is \textit{Strong INcentive Compatible in Expectation} ($\SICIE$) if for any rational player $\playeri$, $\strategyProtocol \strictlyDominates \strategyi$  for all mixed strategies $ \,  \strategyi \in \strategySeti$, with $\strategySeti$ the set of mixed strategies available to $\playeri$, such that  $ \strategyi \, \notequiv \, \strategyProtocol $.
\end{definition}

For a protocol to be $\SICIE$ in the $\ourBAR$ model ensures that all rational players will follow the recommended protocol strategy. However, $\SICIE$ is not on its own sufficient to ensure the safety and liveness of an SMR protocol in $\ourBAR$ model. It is still possible for an adversary to gain more than their fair share of rewards, and as such, increase their total share above the critical threshold of $\adversarialBound$. Towards achieving SMR in the $\ourBAR$ model, it must be ensured that the adversarial share remains strictly bounded by the threshold $\adversarialBound$ required to achieve SMR if all non-adversarial players follow the protocol. We explicitly define what we mean by fairness in the $\ourBAR$ model in Definition \ref{def:fair}.

\begin{definition}\label{def:fair}
   An SMR protocol $\protocol$ with adversary $\adversary$  is \textit{fair} in the $\ourBAR$ model if $P(\share^\round_\adversary\leq \share^\initialisationRound_\adversary) > \approxOne$ for any round $\round \geq\initialisationRound$ .
\end{definition}

With $\SICIE$ and fairness, we have two intuitive properties which turn out to be crucial in achieving $\ourSMR$. In Section \ref{sec:Properties}, we show that it is impossible to guarantee the actions of players controlling more than $1-\adversarialBound$ of the stake if these properties do not hold. Explicitly, we prove that the properties of $\SICIE$ and fairness are necessary, and together sufficient, to achieve $\ourSMR$. 

%% file: ProofsForSICIEFair.tex
\section{Achieving SMR in the \texorpdfstring{$\ourBAR$}{TEXT} Model }\label{sec:Properties}

Towards our final goal of proving that the properties of $\SICIE$ and fairness are necessary, and together sufficient, to achieve $\ourSMR$, the first step is to prove in Lemma \ref{lem:SICIENecessary} that $\SICIE$ is necessary. To allow us to prove this result, we introduce notation which allows us to consider, for a potential SMR protocol, the strategies from which rational players choose. 

\begin{definition}\label{def:strictDominantSet}
    For a rational player $\playeri$ with a set of mixed strategies $\strategySeti$, let $\strategySetNonDominatedi \subseteq \strategySeti$ be such that for all$ \,  \strategyi \in \strategySetNonDominatedi$, there does not exist a $ \strategyi^\utility \in \strategySeti$,  such that $\strategyi^\utility \strictlyDominates \strategyi$. 
\end{definition}

That is, if a mixed strategy $\strategy \in \strategySeti$ is in the set $\strategySetNonDominatedi$, there is no strategy for $\playeri$ which strictly dominates $\strategy$ in expectancy. We provide the following Lemmas towards establishing that rational players will choose strategies exclusively from $\strategySetNonDominatedi$. 

\begin{lemma}\label{lem:dominantXOR}
    For an SMR protocol $\protocol$, a rational player $\playeri$, any strategy $\strategyi^a \in \strategySeti$ , and $|\strategySeti|\geq 2$, either $\strategyi^a \in \strategySetNonDominatedi $ or there is some $ \strategyi^b \in \strategySetNonDominatedi$ such that $\strategyi^b \strictlyDominates \strategyi^a$.
\end{lemma}

\begin{proof}
    We will do this by induction over the cardinalities of $\strategySeti$.
    First we check $| \strategySeti|=2$. If $\strategyi^a$ is in $\strategySetNonDominatedi $, we are finished. Assume otherwise. That is,  $\strategyi^b \strictlyDominates \strategyi^a$, which implies $\strategyi^a \not\strictlyDominates\strategyi^b$, and as such, $\strategyi^b \in \strategySetNonDominatedi $ as required. 
    
    Assume the inductive hypothesis for $| \strategySeti|=k$. 
    
    Now, given this assumption, we must prove our hypothesis holds for $| \strategySeti|=k+1$. Consider a strategy $\strategyi^c \in  \strategySeti$. We need to prove either $\strategyi^c \in \strategySetNonDominatedi$, or there exists $ \, \strategy \in \strategySetNonDominatedi$ with $\strategy \strictlyDominates \strategyi^c$.
    If $\strategyi^c$ is not strictly dominated by any strategy $\strategy \in \strategySeti$, then $\strategyi^c \in \strategySetNonDominatedi$. 
    
    Assume instead there exists some strategy $\strategyi^a \in \strategySeti$, $\strategyi^a \strictlyDominates \strategyi^c$. Consider $Z_i=  \strategySeti \backslash \{ \strategyi^c\}$.
    By the inductive assumption, either $\strategyi^a \in Z_i^\notStrictlyDominated$, or there exists $ \, \strategyi^b \in Z_i^\notStrictlyDominated$ such that $\strategyi^b \strictlyDominates  \strategyi^a$. 
    If $\strategyi^a \in Z_i^\notStrictlyDominated$, then $\strategyi^a \in \strategySetNonDominatedi$, which implies there exists $ \, \strategy \in \strategySetNonDominatedi$ such that $ \strategy \strictlyDominates \strategyi^c$. 
    Otherwise, if $\strategyi^a \not \in Z_i^\notStrictlyDominated$, there exists $ \,  \strategyi^b \in Z_i^\notStrictlyDominated$, with $\strategyi^b \strictlyDominates \strategyi^a$. 
    As $\strategyi^b \strictlyDominates \strategyi^a$, and $\strategyi^a \strictlyDominates \strategyi^c$, this implies $\strategyi^b \strictlyDominates \strategyi^c$. 
    As $\strategyi^b \in Z_i^\notStrictlyDominated$, and $\strategyi^b \strictlyDominates \strategyi^c$, this implies $\strategyi^b \in \strategySetNonDominatedi$. Therefore, there exists $ \,  \strategy \in \strategySetNonDominatedi$ such that $ \strategy \strictlyDominates \strategyi^c$.

\end{proof}

As rational players choose uniformly at random from all mixed strategies which maximise utility, from Lemma \ref{lem:dominantXOR} for a rational player $\playeri$ these mixed strategies will be contained in $\strategySetNonDominatedi$. Moreover, Definition \ref{def:PlayerModel} states that $\playeri$ chooses from these mixed strategies in $\strategySetNonDominatedi$ with uniform probability. Therefore, to ensure rational players follow $\strategyProtocol$ with probability at least $\approxOne$, we must identify the conditions where for any rational player $\playeri$, $\strategySetNonDominatedi = \{ \strategyProtocol \}$. We state this explicitly in Observation \ref{obs:rationalChoices}.

\begin{observation}\label{obs:rationalChoices}
    A rational player $\playeri$ follows $\strategyProtocol$ with probability greater than $\approxOne$ if and only if $\strategySetNonDominatedi = \{ \strategyProtocol \}$. 
\end{observation}

The precise conditions where $\strategySetNonDominatedi = \{ \strategyProtocol \}$ for a rational player $\playeri$ are identified in Lemma \ref{lem:SICIEiffNSD}.

\begin{lemma}\label{lem:SICIEiffNSD}
    For an SMR protocol $\protocol$ and a rational player $\playeri$, $\strategySetNonDominatedi= \{ \strategyProtocol \}$ if and only if $\protocol$ is strong incentive compatible in expectation.
\end{lemma}

\begin{proof}
    If an SMR protocol $\protocol$ is $\SICIE$, then for any rational player $\playeri$,  $\strategyProtocol$ strictly dominates all other strategies in expectation. From Lemma \ref{lem:dominantXOR}, this implies $\strategySetNonDominatedi= \{\strategyProtocol\}$. 
    
    Now we need to show if  $\strategySetNonDominatedi= \{ \strategyProtocol \}$, then $\protocol$ is $\SICIE$. From Lemma \ref{lem:dominantXOR}, we know for any strategy $\strategyi^a$, either $\strategyi^a \in \strategySetNonDominatedi $ or there is some $ \strategyi^b \in \strategySetNonDominatedi$ such that $\strategyi^b \strictlyDominates \strategyi^a$. As the only strategy in $\strategySetNonDominatedi$ is $ \strategyProtocol$, this implies for any strategy $\strategyi^a \, \notequiv \, \strategyProtocol$, $\strategyProtocol \strictlyDominates \strategyi^a$. This implies $\protocol$ is $\SICIE$, as required. 
\end{proof}

\begin{corollary}\label{cor:stategyNotStrictlyDominant}
    For an SMR protocol $\protocol$ and a rational player $\playeri$, $P(\playeri \, \text{chooses} \, \strategyProtocol) > \approxOne$ if and only if $\protocol$ is strong incentive compatible in expectation.
\end{corollary}

\begin{proof}
    Follows from Observation \ref{obs:rationalChoices} and Lemma \ref{lem:SICIEiffNSD}.
\end{proof}

This allows us to prove $\SICIE$ is a necessary property to achieve $\ourSMR$.

\begin{lemma}\label{lem:SICIENecessary}
    For an SMR protocol $\protocol$, if $\protocol$ is not strong incentive compatible in expectation, then $\protocol$ fails in the $\ourBAR$ model.
\end{lemma}

\begin{proof}
    Consider such a protocol $\protocol$. As a consequence of not $\SICIE$, for a rational player $\playeri$, this means $P(\playeri \, \text{chooses} \, \strategyProtocol)$ is not greater than $\approxOne$, applying Corollary \ref{cor:stategyNotStrictlyDominant}. 
    From Definition \ref{def:SMRFails} we are required to consider  $\share^\initialisationRound_\adversary$ maximal. Given this rational $\playeri$ and a maximal adversary, there is now players controlling greater than or equal to $\adversarialBound$ of the total stake who will not choose a strategy equivalent to $\strategyProtocol$ with non-negligible probability in $\kappa$. Using the notation of Definition \ref{def:SMRFails}, this means $p^\round_{\protocol}$ is not greater than $\approxOne$ for some $\round \geq \initialisationRound$, which implies $\protocol$ fails in the $\ourBAR$ model.
\end{proof}

Using similar arguments, we are able to prove fairness is also necessary for a protocol to achieve $\ourSMR$.

\begin{lemma}\label{lem:fairNecessary}
    For an SMR protocol $\protocol$, if $\protocol$ is not fair then $\protocol$ fails in the $\ourBAR$ model. 
\end{lemma}

\begin{proof}
    If $\protocol$ is not fair, there exists $ \, \round \geq \initialisationRound$ such that $P(\share^\round_\adversary > \share^\initialisationRound_\adversary)$ is not negligible in $ \kappa$. From Definition \ref{def:SMRFails}, we are required to consider the case where $\share^\initialisationRound_\adversary$ is maximal. In this case, the probability that the adversary controls greater than or equal to $\adversarialBound$ of the stake at round $\round$ is is non-negligible in $ \kappa$ given $P(\share^\round_\adversary > \share^\initialisationRound_\adversary)$ is non-negligible in $ \kappa$.  
    Given the uniform strategy selection probability of Byzantine players across all possible strategies, this implies that $p^\round_{\protocol}$ is not greater than $\approxOne$. 
    Therefore, $\protocol$ fails in the $\ourBAR$ model. 
\end{proof}

Collecting the results of this section, with some additional proof-work, we are equipped to prove the main theorem of the paper, Theorem \ref{thm:fair+SICIE=Necessary+Sufficient}. 

\begin{theorem}\label{thm:fair+SICIE=Necessary+Sufficient}
    For an SMR protocol $\protocol$, $\protocol$ achieves $\ourSMR$ if and only if $\protocol$ is strong incentive compatible in expectation and fair. 
\end{theorem}

\begin{proof}
    For an SMR protocol $\protocol$, we will first prove that if $\protocol$ achieves $\ourSMR$ then $\protocol$ is $\SICIE$ and fair.
    Using the contrapositive of Lemma \ref{lem:SICIENecessary}, we have that if $\protocol$ achieves $\ourSMR$ (does not fail in the $\ourBAR$ model), then $\protocol$ is $\SICIE$. 
    Similarly, using the contrapositive of Lemma \ref{lem:fairNecessary}, we have that if  $\protocol$ achieves $\ourSMR$, then $\protocol$ is fair.
    
    We now need to prove if $\protocol$ is $\SICIE$ and fair then $\protocol$ achieves $\ourSMR$.
    By $\SICIE$ and Corollary \ref{cor:stategyNotStrictlyDominant}, this implies all rational players will always choose $\strategyProtocol$. Furthermore, as $\protocol$ is fair, from Definition \ref{def:fair}, we know rational players will maintain greater than $1-\adversarialBound$ of the stake in every round with probability greater than $\approxOne$. Therefore, we have players controlling greater than $1-\adversarialBound$ of the stake who will follow $\strategyProtocol$ with probability greater than $ \approxOne$, which is precisely the definition of $\protocol$ achieving $\ourSMR$ from Definition \ref{def:SMRFails}.
\end{proof}

This crucial theorem completes the first part of the paper, identifying the properties of $\SICIE$ and fairness as both necessary, and together sufficient, for a protocol to achieve $\ourSMR$, independently of network assumptions and adversarial capabilities. We now proceed to outline the $\protocolName$ protocol, demonstrating that it is possible to satisfy $\SICIE$ and fairness in the $\ourBAR$ model.

%% file: Protocol/Protocol.tex
\newcommand{\Lim}[1]{\raisebox{0.5ex}{\scalebox{0.8}{$\displaystyle \lim_{#1}\;$}}}

\section{\texorpdfstring{$\protocolName$}{TEXT}}\label{sec:Protocol}

In this section, we provide the encoding of $\protocolName$, and give an overview of the main differences between the $\protocolName$ protocol and Tendermint. We assume a partially synchronous network communication model as in Tendermint \cite{Tendermint}. Players are connected to nodes in a dynamic wide area network, with each node having direct connections to a subset of all other nodes, forming a sparsely connected graph of communication channels between nodes. Non-Byzantine player messages are transmitted through \textit{gossiping}; players send a message to neighbouring nodes, who echo messages to their neighbours until all nodes eventually receive the message. Formally, there is global stabilisation round $\GST>0$, such that all messages sent at round $\sendTime>0$ are delivered by round $\deliverTime= max(\sendTime, \GST) +\Delta$ for some unknown number of rounds $\Delta>0$. 

We assume rational players are aware that there is a fixed, but unknown, upperbound $\Delta$ on message delivery between players in synchrony, which we refer to as $\Delta$-synchrony, but are unaware of how many players are in $\Delta$-synchrony at any given time. For a message $m$, a call to $\Broadcast(m)$ sends a $m$ to all players, including oneself, under the same gossiping specification. This is partial synchrony as defined in \cite{PartialSynchronyDwork}.

\input{Protocol/ThreatModel}

\input{Protocol/AlgorithmicSetup}

\input{Protocol/ProtocolEncoding}

\input{Protocol/ProtocolDescription}

%% file: Protocol/ThreatModel.tex
\subsection{Threat Model}\label{sec:ThreatModel}

In $\protocolName$, protocol actions take negligible amounts of time compared to network delays, so if all non-Byzantine players behave correctly and receive the same sequence of messages their machines will be in the same state. Rational players ignore messages which have not been signed using a protocol-associated private key.

We consider an adversary $\adversary$ with the following properties:

\begin{enumerate}

    \item $\adversary$ can read all messages sent by non-Byzantine parties, but cannot existentially forge signatures. 
    
    \item $\adversary$ can control and coordinate all Byzantine players in any way, with unknown utility function. 
    
    \item  At initialisation we have $\partialSyncAdvLimit-\delta <\share^1_\adversary<\partialSyncAdvLimit=\adversarialBound$, for some $\delta >0$, in line with the partially synchronous network distribution limits \cite{PartialSynchronyDwork}. 
    
    \item At initialisation, $\adversary$ can choose to corrupt any  $1 \leq \numByzantines<\numPlayers-2$ players, say $\player_{1}, ..., \player_{\numByzantines}$ with shares $\share^\initialisationHeight_{1}, ..., \share^\initialisationHeight_{\numByzantines}$, such that $\sum^{\numByzantines}_{i=1}\share^\initialisationHeight_i=\share^\initialisationHeight_\adversary$. 
    
    \item Given $\adversary$ corrupts players $\player_1,..., \player_\numByzantines $ as Byzantine for consensus on a block at height $\blockchainHeight$ with shares $ \share_1^{\blockchainHeight-1},...,$ $ \share_\numByzantines^{\blockchainHeight-1} $, the adversarial share at the proceeding height is calculated as $\shareH_\adversary=\sum^\numByzantines_{i=1}\shareH_i$.

\end{enumerate}

\begin{remark}\label{rem:StaticAdv}
    In this work, we focus on static adversaries. It is possible to extend our results to an adaptive adversary who can re-select the set of Byzantine players after every decision. To do so requires significant additional code and further assumptions that preserve adversarial stake throughout corruptions. We choose to leave this as future work, as it only stands to detract from the primary focus of the paper, that is, to demonstrate the importance of ByRa SMR and how it can be achieved in real-world protocols with $\protocolName$. 
\end{remark}

%% file: Protocol/AlgorithmicSetup.tex
\algdef{SE}[EVENT]{Upon}{EndUpon}[1]{\textbf{upon}\ #1\ \algorithmicdo}{\algorithmicend}\algtext*{EndUpon}

\newcommand{\newconstruct}[5]{%
  \newenvironment{ALC@\string#1}{\begin{ALC@g}}{\end{ALC@g}}
   \newcommand{#1}[2][default]{\ALC@it#2\ ##2\ #3%
     \ALC@com{##1}\begin{ALC@\string#1}}
   \ifthenelse{\boolean{ALC@noend}}{
     \newcommand{#4}{\end{ALC@\string#1}}
   }{
     \newcommand{#4}{\end{ALC@\string#1}\ALC@it#5}
   } 
}

\newcommand\Disseminate{\textbf{Disseminate}}

\newcommand\Proposal{\mathsf{PROPOSAL}}
\newcommand\ProposalPart{\mathsf{PROPOSAL\mbox{-}PART}}
\newcommand\PrePrepare{\mathsf{INIT}} \newcommand\Prevote{\mathsf{PREVOTE}}
\newcommand\Precommit{\mathsf{PRECOMMIT}}
\newcommand\Decision{\mathsf{DECISION}}
\newcommand\Slash{\mathsf{SLASH}}

\newcommand\ViewChange{\mathsf{VC}}
\newcommand\ViewChangeAck{\mathsf{VC\mbox{-}ACK}}
\newcommand\NewPrePrepare{\mathsf{VC\mbox{-}INIT}}
\newcommand\coord{\mathsf{proposer}}

\newcommand\newHeight{newHeight} \newcommand\newRound{newRound}
\newcommand{\nil}{nil} \newcommand\id{id} 
\newcommand{\assign}{\leftarrow}
\newcommand{\inc}[1]{#1 \assign #1 + 1}
\newcommand{\isdef}{:=}
\newcommand{\ident}[1]{\mathit{#1}}
\def\newident#1{\expandafter\def\csname #1\endcsname{\ident{#1}}}

\newcommand{\eg}{{\it e.g.}}
\newcommand{\ie}{{\it i.e.}}
\newcommand{\apriori}{{\it apriori}}
\newcommand{\etal}{{\it et al.}}

\newcommand{\propose}{\textit{propose}}
\newcommand\prevote{\textit{prevote}}
\newcommand\prevoteWait{prevoteWait}
\newcommand\precommit{\textit{precommit}}
\newcommand\precommitWait{precommitWait}
\newcommand\commit{\textbf{\textit{commit}}}
\newcommand{\upon}{\textbf{upon}}

\newcommand\timeoutPropose{timeout}
\newcommand\timeoutPrevote{timeout}
\newcommand\timeoutPrecommit{timeout}
\newcommand\proofOfLocking{proof\mbox{-}of\mbox{-}locking}
\newcommand{\SPACE}{\vspace{3mm}}
\newcommand{\SHORTSPACE}{\vspace{1mm}}

\newcommand{\playerHeight}{\heightVariable_i}
\newcommand{\playerEpoch}{\epoch_i}
\newcommand{\playerStep}{\step_i}
\newcommand{\playerBlockchain}{\blockchain_i}
\newcommand{\lockedValue}{\textit{lockValue}_i}
\newcommand{\lockedEpoch}{\textit{lockEpoch}_i}
\newcommand{\validValue}{\textit{validValue}_i}
\newcommand{\validEpoch}{\textit{validEpoch}_i}
\newcommand{\proposerValidEpoch}{\textit{validEpoch}}
\newcommand{\prevoteProof}{\textit{prevoteProof}_i}
\newcommand{\precommitProof}{\textit{precommitProof}_i}
\newcommand{\proposedValue}{\textit{proposal}_i}
\newcommand{\logicalAnd}{\textbf{and}}
\newcommand{\logicalOr}{\textbf{or}}

%% file: Protocol/ProtocolEncoding.tex
\begin{algorithm}
\def\baselinestretch{1} \scriptsize \raggedright
\caption{$\protocolName$ protocol for a player $\playeri$}
\label{alg:tenderstake} 
	\begin{algorithmic}[1] 
	
		\Function{$Initialise(Genesis)$}{}\label{func:Initialise}
		    
		    \State $\playerBlockchain := [\textit{Genesis}]$ \Comment{$\playeri$'s blockchain as a vector} \label{line:blockchainDef}
		    \State $\playerHeight := \initialisationHeight$  \Comment{Tracks height of $\playerBlockchain$} 
		    \State $\playerEpoch := \initialisationEpoch$  
		    \State $\playerStep  \in \{\propose, \prevote, \precommit \}$ \label{line:variableDefStart}
		    \State $\lockedValue := nil$ 
		    \State $\lockedEpoch := -1$ 
		    \State $\validValue := nil$ 
		    \State $\validEpoch := -1$ 
		    \State $\playerStake := \textit{Genesis.stake}() $ \Comment{Total stake}\label{algo:Initialise:line:setStake}
		    \State $\playerShares := \textit{Genesis.shares}()  $ \Comment{Vector of player shares}\label{algo:Initialise:line:setShares}
		    \State $\playerReward :=\textit{ Genesis.reward}()$ \Comment{Per-Block reward}
		    \State $\playerDeviationProofs := [ nil \ for \ j \ \in \playerIndices ]$ \Comment{Deviation proofs} \label{line:variableDefEnd}
            \State $\prevoteProof := nil$ 
		    \State $\precommitProof := nil$

		\EndFunction
		
		\SHORTSPACE 
		\State \textbf{upon} start \textbf{do} $\textit{StartEpoch}(\initialisationEpoch)$ 
		
		\SHORTSPACE 
		\Function{ $StartEpoch$}{$\epoch$}  \label{line:tab:startEpoch} 
		\State $\playerEpoch \assign \epoch$ 
		\State $\playerStep \assign \propose$ 
		\If{$\coord(\playerHeight, \playerEpoch) = \playeri$}
		\If{$\validValue \neq \nil$} \label{line:tab:isThereLockedValue}
		\State $\proposedValue \assign \validValue$ 
		\Else 
		\State $\proposedValue \assign \textit{getValue}().\textit{include}( \playerDeviationProofs)$ 
		\label{line:tab:getValidValue} 
		\EndIf 	  
		\State \Broadcast $\langle\Proposal,\playerHeight, \playerEpoch, \proposedValue, \validEpoch,\prevoteProof \rangle$
		\label{line:tab:send-proposal} 
		\Else 
		\State \textbf{schedule} $\textit{OnTimeoutPropose}(\playerHeight,
		\playerEpoch)$ to be executed \textbf{after} $\timeoutPropose()$ \label{line:timeoutPropose}
		\EndIf
		\EndFunction

		\SPACE 
		\Upon {$\langle\Proposal,\playerHeight,\playerEpoch, \val, -1,\textit{proof}\rangle \ \With \ \textit{valid}(\textit{proof})$ \From\ $\coord(\playerHeight,\playerEpoch)$ \WhileLocal\ $\playerStep = \propose$  } \label{line:tab:recvProposal}
			\If{$\textit{valid}(\val) \ \logicalAnd \ (\lockedEpoch = -1  \ \logicalOr \ \lockedValue = \val$)}
			\label{line:tab:accept-proposal-2} 
			    
				\State \Broadcast \ $\langle\Prevote,\playerHeight,\playerEpoch,\val, \prevoteProof\rangle$  
				\label{line:tab:prevote-proposal}	
			\Else
			\label{line:tab:acceptProposal1}		
				\State \Broadcast \ $\langle\Prevote,\playerHeight,\playerEpoch,\nil, \prevoteProof\rangle$  
				\label{line:tab:prevote-nil}	
			\EndIf
				\State $\playerStep \assign \prevote$ \label{line:tab:setStateToPrevote1} 
		\EndUpon

		\SPACE 
		\Upon{$\langle\Proposal,\playerHeight,\playerEpoch, \val, \proposerValidEpoch,\textit{proofProposal}\rangle$ \With\ $\textit{ valid}(\textit{proofProposal}) $ \From\ $\coord(\playerHeight,\playerEpoch)$
			\logicalAnd $\ > \partialSyncMajority$ $\langle\Prevote,\playerHeight, \proposerValidEpoch,\val,\textit{proofPrevote}\rangle \ $ \With \ $ \textit{valid}(\textit{proofPrevote})$ \WhileLocal\ \big($\playerStep = \propose \ \logicalAnd \ (\proposerValidEpoch \ge 0 \ \logicalAnd \ \proposerValidEpoch < \playerEpoch) \big)$}
		\label{line:tab:acceptProposal} 
		\State $\prevoteProof \assign $ $ \textit{proof}( > \partialSyncMajority$ $\langle\Prevote,\playerHeight, \proposerValidEpoch, \val \rangle \cup \prevoteProof) $ \label{line:prevoteEx} 
		\If{$\textit{valid}(\val) \ \logicalAnd \ (\lockedEpoch < \proposerValidEpoch
			\ \logicalOr \ \lockedValue = \val)$} \label{line:tab:cond-prevote-higher-proposal}	
			
			\State \Broadcast \ $\langle\Prevote,\playerHeight,\playerEpoch,\val, \prevoteProof\rangle$
			\label{line:tab:prevote-higher-proposal}		 
		\Else
			\label{line:tab:acceptProposal2}		
			\State \Broadcast \ $\langle\Prevote,\playerHeight,\playerEpoch,\nil, \prevoteProof\rangle$  
			\label{line:tab:prevote-nil2}	
		\EndIf
		\State $\playerStep \assign \prevote$ \label{line:tab:setStateToPrevote3} 	 
		\EndUpon

		\SPACE 
		\Upon{$> \partialSyncMajority$ $\langle\Prevote,\playerHeight, \playerEpoch,*,\textit{proof}\rangle$ \With $\ \textit{valid}(\textit{proof})$ \WhileLocal\ $\playerStep = \prevote$ for the first time}
		\label{line:tab:recvAny2/3Prevote} 
		    \State $\precommitProof \assign \  \textit{proof}( > \partialSyncMajority \langle\Prevote,\playerHeight, \playerEpoch,*\rangle )$ \label{line:precommitEx} 
		    \State \textbf{schedule} $\textit{OnTimeoutPrevote}(\playerHeight, \playerEpoch)$ to be executed \textbf{after}  $\timeoutPrevote()$ \label{line:tab:timeoutPrevote} 
		\EndUpon

		\SPACE 
		\Upon{$\langle\Proposal,\playerHeight,\playerEpoch, \val, *, \textit{proofProposal}\rangle$  \ \With \ $ \textit{valid}(\textit{proofProposal})$ \From\ $\coord(\playerHeight,\playerEpoch)$ \logicalAnd \ $> \partialSyncMajority$ $\langle\Prevote,\playerHeight, \playerEpoch,\val,\textit{proofPrevote}\rangle \ \With \ \textit{valid}(\textit{proofPrevote})$ \WhileLocal\ $\textit{valid}(\val) \ \logicalAnd \ \playerStep \ge \prevote$ for the first time  }
		\label{line:tab:recGoodPrevotes} 
		\If{$\playerStep = \prevote$}	\label{line:tab:goodPrevotesAndStepPrevote}
			\State $\lockedValue \assign \val$                \label{line:tab:setLockedValue}
			\State $\lockedEpoch \assign \playerEpoch$   \label{line:tab:setLockedEpoch} 
			\State $\precommitProof \assign \ \textit{proof}(> \partialSyncMajority$ $\langle\Prevote,\playerHeight, \playerEpoch,\val\rangle)$
			\State \Broadcast \ $\langle\Precommit,\playerHeight,\playerEpoch,\val, \precommitProof\rangle$  
			\label{line:tab:precommit-v}	
			\State $\playerStep \assign \precommit$ \label{line:tab:setStateToCommit} 
		\EndIf 
		\State $\validValue \assign \val$ \label{line:tab:setValidEpoch} 
		\State $\validEpoch \assign \playerEpoch$ \label{line:tab:setValidValue} 
		\EndUpon

		\SHORTSPACE 
		\Upon{$> \partialSyncMajority$ $\langle \Prevote,\playerHeight,\playerEpoch, \nil, \textit{proof} \rangle  \ \With \ \textit{valid}(\textit{proof})$ \WhileLocal\ $\playerStep = \prevote$} \label{line:nilPrevotes}
			\State $\precommitProof \assign \ \textit{proof}(> \partialSyncMajority$ $\langle\Prevote,\playerHeight,\playerEpoch,\nil\rangle )$
			\State \Broadcast \ $\langle\Precommit,\playerHeight,\playerEpoch, \nil,\precommitProof\rangle$
			\label{line:tab:precommit-v-1} 
			\State $\playerStep \assign \precommit$ 
		\EndUpon

	\end{algorithmic} 
	
\end{algorithm}

\addtocounter{algorithm}{-1}
\begin{algorithm}
\def\baselinestretch{1} \scriptsize\raggedright
\caption{$\protocolName$ protocol (ctd.)}
	\label{alg:tenderstakectd} 
	\begin{algorithmic}[1]		
		
	\setcounter{ALG@line}{56}
	    \SPACE 
		\Upon{$> \partialSyncMajority$ $\langle\Precommit,\playerHeight,\playerEpoch,*, \textit{proof}\rangle  \ \With \ \textit{valid}(\textit{proof})$ for the first time} 
		\label{line:tab:recvAny2/3Precommit} 
		    \State $\prevoteProof \assign \ \textit{proof}( > \partialSyncMajority$ $\langle\Precommit,\playerHeight,\playerEpoch,*\rangle)$
			\State \textbf{schedule} $\textit{OnTimeoutPrecommit}(\playerHeight, \playerEpoch)$ to be executed \textbf{after} $\timeoutPrecommit()$ \label{line:timeoutPrecommit}
		\EndUpon

	    \SPACE 
		\Upon{$> \partialSyncMajority$ $\langle \Precommit,\playerHeight,\playerEpoch, \nil, \textit{proof} \rangle  \ \With \ \textit{valid}(\textit{proof})$ \WhileLocal\ $\playerStep = \precommit$}  \label{line:FailedEpoch} 
			\State $\textit{StartEpoch}(\playerEpoch + 1)$
		\EndUpon

		\SPACE 
		\Upon{$\langle\Proposal,\playerHeight,\epoch, \val, *, \textit{proofProposal}\rangle  \ \With \ \textit{valid}(\textit{proofProposal}) $ \From\ $\coord(\playerHeight,\epoch) $  \logicalAnd $\ > \partialSyncMajority$ $\langle\Precommit,\playerHeight,\epoch,\val, \textit{proofPrecommit}\rangle  \ \With \ \textit{valid}(\textit{proofPrecommit}) $} 
		\label{line:tab:onDecideRule} 
		    \State $\textit{newDeviators} \assign \val.\textit{deviators}() \backslash \playerBlockchain.\textit{deviators}() $
			\If{$\textit{valid}(\val)$} \label{line:validDecisionValue} 
			    \If{$|\textit{newDeviators}|>0 $ } \Comment{True if deviators in $\val$ not in $\playerBlockchain$} \label{line:conditionNewDevs}
			        \State $\textit{adjustForSlashing}( \textit{newDeviators})$
			        
			    \EndIf
			    \State $\prevoteProof \assign \ \textit{proof}( > \partialSyncMajority \langle\Precommit,\playerHeight,\epoch,\val\rangle)$
				\State $\playerBlockchain.\textit{append}(\val.include(\prevoteProof))$ \label{line:addBlock} 
				\State $\playerStake \assign \playerStake \ +  \playerReward$ \label{line:addReward} 
				\State $\playerHeight \assign \playerHeight + 1$ \label{line:increaseHeight} 
				\State reset $\lockedEpoch$, $\lockedValue$, $\validEpoch$,  $\validValue$ to initial values
				\State $\textit{StartEpoch}(1)$   	
			\EndIf 
		\EndUpon

		\SHORTSPACE 
		\Upon{$> \partialSyncAdvLimit$ $\langle*,\playerHeight,\textit{epoch}, *, \textit{proof}\rangle$ \With $\ (\textit{epoch} > \playerEpoch$  \logicalAnd $ \ valid(proof))$}
		\label{line:tab:skipEpochs} 
		    \State $\prevoteProof \assign \ \textit{proof}(> \partialSyncAdvLimit$ $\langle*,\playerHeight,\textit{epoch},*,*\rangle )$
			\State $\textit{StartEpoch}(\textit{epoch})$ \label{line:tab:nextEpoch2} 
		\EndUpon
		
		\SHORTSPACE 
		\Function{$OnTimeoutPropose(height,epoch)$}{} \label{line:tab:onTimeoutPropose} 
		\If{$\textit{height} = \playerHeight \ \logicalAnd \ \textit{epoch} = \playerEpoch \ \logicalAnd \ \playerStep = \propose$} 
		    
			\State \Broadcast \ $\langle\Prevote,\playerHeight,\playerEpoch, \nil,\prevoteProof\rangle$ 
		 	\label{line:tab:prevote-nil-on-timeout}	
		 	\State $\playerStep \assign \prevote$ 
		 \EndIf	
		 \EndFunction
		
		\SHORTSPACE 
		\Function{$OnTimeoutPrevote(height,epoch)$}{}  \label{line:tab:onTimeoutPrevote} 
		\If{$\textit{height} = \playerHeight \ \logicalAnd \ \textit{epoch} = \playerEpoch \ \logicalAnd \ \playerStep = \prevote$} 
		    
			\State \Broadcast \ $\langle\Precommit,\playerHeight,\playerEpoch,\nil,\precommitProof\rangle$   
			\label{line:tab:precommit-nil-onTimeout}
			\State $\playerStep \assign \precommit$ 
		\EndIf	
		\EndFunction
		
		\SHORTSPACE 
		\Function{$OnTimeoutPrecommit(height,epoch)$}{}  \label{line:tab:onTimeoutPrecommit} 
		\If{$\textit{height} = \playerHeight \ \logicalAnd \ \textit{epoch} = \playerEpoch$}
			\State $\textit{StartEpoch}(\playerEpoch + 1)$ \label{line:tab:nextEpoch} 
		\EndIf
		\EndFunction

		\SHORTSPACE 
		\Upon{$m$ \From \ $\player^j$ \With \ $\textit{valid}(m)= \textit{FALSE}$ for the first time } \label{line:deviatingStart}
		    \State $\playerDeviationProofs[j] \assign \textit{proof}(\textit{valid}(m)= \textit{FALSE})$
			\State \Broadcast \ $\langle \Slash,\player^j,\playerHeight, \playerEpoch, m, \playerDeviationProofs[j]\rangle$ \label{line:deviationIdentified}
		\EndUpon

		\SHORTSPACE 
		\Upon{$\langle \Slash,\player^j, m, \textit{proof}\rangle$ \ \From \ $\player^k \ \With \ \textit{valid}(\textit{proof})$}
		        \If{$ \playerDeviationProofs[j] = nil$}
		            \State $\playerDeviationProofs[j] \assign \textit{proof}$
		            \State \Broadcast \ $\langle \Slash,\player^j, m, \playerDeviationProofs[j]\rangle$ \label{line:deviationEchoed}	 
		        \EndIf
		\EndUpon	
			
		\SHORTSPACE 
		\Function{$adjustForSlashing(newDeviators)$}{}  \label{line:removeSlashers}
		    \State $\slashedShare \assign sum(\playerShares[newDeviators])$	
		    \State $\playerShares[newDeviators] \assign 0 $	 \label{line:destroyDeviatorStake}
		    \State $\playerShares \assign \big[ \frac{\playerShares[k]}{1-\slashedShare} \ \textbf{for} \ k \ \in [1,...,\numPlayers] \big]$ \label{line:adjustShares}
		    \State $\playerReward \assign (1-\slashedShare) \playerReward$ \Comment{Same reward adjustment as in FAIRSICAL} \label{line:adjustReward}
		    \State $\playerStake \assign (1-\slashedShare)\playerStake   $ \Comment{Remove deviating stake} \label{line:adjustStake}
		    \State $\playerStake \assign \playerStake  +  \textit{sum}([\textit{Genesis.shares}()[j] \ \textbf{for} \ j \in newDeviators]) \cdot \playerReward$  \label{line:splitDeviatorPot} \Comment{Slash Bonus, not dependant on ordering} 
		\EndFunction

	\end{algorithmic} 
\end{algorithm}

%% file: Protocol/ProtocolDescription.tex
\subsection{Protocol Outline}\label{sec:protocolDescription}

We now describe the pseudocode of $\protocolName$ as outlined in Algorithm \ref{alg:tenderstake}. As the goal of Section \ref{sec:Protocol} is to amend Tendermint to achieve $\ourSMR$, readers of \cite{LatestGossipTendermint} will notice that we use large parts of the code and descriptions from that work. We describe the entire code here for completeness, and highlight the differences in $\protocolName$ to Tendermint as they arise. The two fundamental additions to the Tendermint protocol used by $\protocolName$ are proof-of-transition and slashing functionalities, described in detail in Sections \ref{sec:PoT} and \ref{sec:Slash} respectively, and included in the code of Algorithm \ref{alg:tenderstake}. Proof-of-transition ensures players who send a message at a particular height/ epoch/ step have gotten there by following the protocol, while the slashing functionality enforces the use of proofs-of-transition, as well as the sending of valid messages in general, by punishing players for sending invalid messages.

In Tenderstake, every correct player is initialised by passing a block \textit{Genesis} to the \textit{Initialise} function (line \ref{func:Initialise}). This ensures all players start from a common state. Block \textit{Genesis} contains information on player shares, stake and per-block reward at initialisation. 

The algorithm is presented as a set of $\upon$ rules that are to be executed automatically once the corresponding logical condition is \textit{TRUE}. Variables with sub-index $i$ denote player $\playeri$'s local state variables, while those without are value placeholders. The sign $*$ denotes any value. We use the convention of 
$>\frac{x}{3} \ m \  \With \ \textit{COND}$
to stand for the logical statement which is \textit{TRUE} if and only if players controlling more than $\frac{x}{3}$ of the total stake with respect to $\playeri$'s blockchain $\blockchain_i$ (represented as a vector in line \ref{line:blockchainDef}) deliver messages, with each message $m$ satisfying the logical condition \textit{COND}. If $m$ contains a proposed deviator, that deviator's share does not count towards the tally (it would be irregular that a player would affirm a message which tried to destroy their own stake). 

As the total voting power in the system is 1, this means if there are new deviators proposed in a particular valid value (line \ref{line:tab:getValidValue}) with total share $\share_{dev}$, the maximum total voting share for that value is $1-\share_{dev}$. This ensures any player $\playeri$ at height $\blockchainHeight$ with share $\share_i^{\blockchainHeight-1}$ after deciding on the block at height $\blockchainHeight-1$ can only have 0 or $\share_i^{\blockchainHeight-1}$ voting power. 
Rules ending with `for the first time' should only be executed on the first time the corresponding condition is \textit{TRUE}. 

The algorithm proceeds in epochs, with each epoch having a dedicated proposer. The mapping of epochs to proposers is known to all players, with the function $\coord(h, \ \textit{epoch})$ returning the proposer for epoch $\textit{epoch}$ given current blockchain height $h$. Player state transitions are triggered by message reception and by expiration of the timeout function $\textit{timeout}()$. Timeouts are to be called once per step during each epoch, and only trigger a transition if the player has not updated their step or epoch variable since starting the timeout function. 

In \cite{LatestGossipTendermint} it is proved that non-Byzantine players need to incorporate increasing timeouts in the number of epochs at a particular height to guarantee eventual progression. In $\protocolName$, we also incorporate increasing timeouts in epochs, but instead leave the precise definition of the timeout function $\textit{timeout}()$ to each player. We do however place the following restriction on the $\textit{timeout}()$ calculation: the value of $\textit{timeout}()$ is increasing in the number of epochs at every height, such that $ \Lim{\epoch \to\infty} \textit{timeout}(\epoch) \rightarrow \infty$.

The intuition behind this choice is leaving it sufficiently general so as to not risk choosing some specific delta/ function for delta which would expose us to unnecessary optimisation analysis, while also ensuring $\protocolName$ retains the property of increasing timeouts in the number of epochs at each height required in the original Tendermint protocol to guarantee safety and liveness.

Messages in $\protocolName$ contain one of the following tags: \\ $\Proposal, $ $ \Prevote, $ $ \Precommit$ and $\Slash$. The $\Proposal$ tag is used by the proposer of the current epoch to suggest a potential decision value (line \ref{line:tab:send-proposal}), while $ \Prevote $ and $\Precommit$ are votes for a proposed value, as in Tendermint. $\Slash$ messages identify player deviations, and are described in detail in Section \ref{sec:Slash}.

Every player $\playeri$ stores the following variables in the $\protocolName$ protocol: $ \playerStep, $ $  \ \lockedValue, $ $ \lockedEpoch, $ $ \validValue, $ $ \validEpoch, $ \\ $\playerStake , $ $ \playerShares,$ $\playerReward,$ and $\playerDeviationProofs $, initialised in lines \ref{line:variableDefStart}-\ref{line:variableDefEnd}. The $\playerStep$ tracks the current step of the protocol execution during the current epoch. The $\lockedValue$ stores the most recent value for which a  $\Precommit$ message was sent by $\playeri$ for a non-$\nil$ value, with $\lockedEpoch$ the epoch in which $\lockedValue$ was updated. As $\playeri$ can only decide on a value $\val$ if more than $ \partialSyncMajority$ voting power equivalent $\Precommit$ messages are received for $\val$, possible decision values can be any value locked by more than $ \partialSyncAdvLimit$ voting power equivalent players. Therefore any value $\val$ for which $\Proposal$ and more than $ \partialSyncMajority$ voting power equivalent $\Prevote$ messages are received in some epoch is a possible decision value. The $\validValue$ stores this value, while $\validEpoch$ stores the epoch where this update occurred. The $\playerStake$ tracks the total stake in the system, and $\playerShares$ the current player shares of $\playerStake$. The $\playerReward$ is the total reward to be distributed among all players for deciding on the next value in $\playerBlockchain$. The $\playerDeviationProofs$ vector tracks locally observed deviators as identified by $\Slash$ messages. 

\input{Protocol/ProofOfTransitionFunctionailty}

\input{Protocol/SlashingFunctionality}
\input{Protocol/ProofofDeviation}

\subsubsection{Life-Cycle of an Epoch}

Every epoch starts by a proposer suggesting a value in a $\protocolName$ message (line \ref{line:tab:send-proposal}). If $\validValue = \nil$, this proposed value is generated by the external $\textit{getValue}()$ function (line \ref{line:tab:getValidValue}), as in Tendermint. In $\protocolName$, players also include in their newly generated propose values any deviation proofs they have received that are not currently in $\blockchain_i$. Otherwise if $\validValue \neq \nil$, the proposer proposes $\validValue$. The proposer attaches $\validEpoch$ to the message so other processes are informed of the last epoch in which the proposer observed $\validValue$ as a possible decision value.

Upon receiving a valid $\langle\Proposal,\playerHeight,\playerEpoch, $ $ \val,$ $ \proposerValidEpoch,$ $\proofProposal \rangle$ message, a correct player $\playeri$ accepts the proposed value $\val$ if both the external function $\textit{valid}(\val)$ returns $\textit{TRUE}$ and either $\playeri$ has not locked any value ($\lockedEpoch= -1$) or $\playeri$ has locked on $\val$ (line \ref{line:tab:accept-proposal-2}). For a valid proposed value $\val$ with $\proposerValidEpoch \geq 0$, if $\proposerValidEpoch > \validEpoch$ (the proposed value was more recent than $\playeri$'s locked value) or $\lockedValue= \val$, $\playeri$ will accept $\val$ (line \ref{line:tab:cond-prevote-higher-proposal}). Otherwise, $\playeri$ rejects the proposal by sending a $\Prevote$ message for $\nil$. $\playeri$ will also send a $\Prevote$ message for $\nil$ if the timeout triggered in line \ref{line:timeoutPropose} expires and they have not sent a $\Prevote$ message for any other value during this epoch yet (line \ref{line:tab:prevote-nil-on-timeout}). 

If a correct player $\playeri$ receives a $\Proposal$ message for a valid value $\val$ and $\Prevote$ messages for $\val$ from players controlling more than $\partialSyncMajority$ of the share as described by $\val$, then it sends a $\Precommit$ message for $\val$. Otherwise, they send a $\Precommit$ message for $\nil$. A correct process will also send a $\Precommit$ message for $\nil$ if the timeout triggered in line \ref{line:tab:timeoutPrevote} expires and they have not sent a $\Precommit$ message for their current epoch yet (line \ref{line:tab:precommit-nil-onTimeout}). A correct player decides on a value $\val$ if it receives in some epoch $\epoch$ a $\Proposal$ message for $\val$ and $\Precommit$ messages for $\val$ from players controlling more than $\partialSyncMajority$ of the share as described by $\val$. On a decision, $\val$, including proof of the $\Precommit$ messages allowing $\playeri$ to decide on $\val$, are appended to $\playerBlockchain$ (line \ref{line:addBlock}). Otherwise, to ensure progression, if the timeout triggered at line \ref{line:timeoutPrecommit} expires, the player proceeds to the next epoch (line \ref{line:tab:nextEpoch}).

\begin{figure}

\centering
\includegraphics[width=0.45\textwidth]{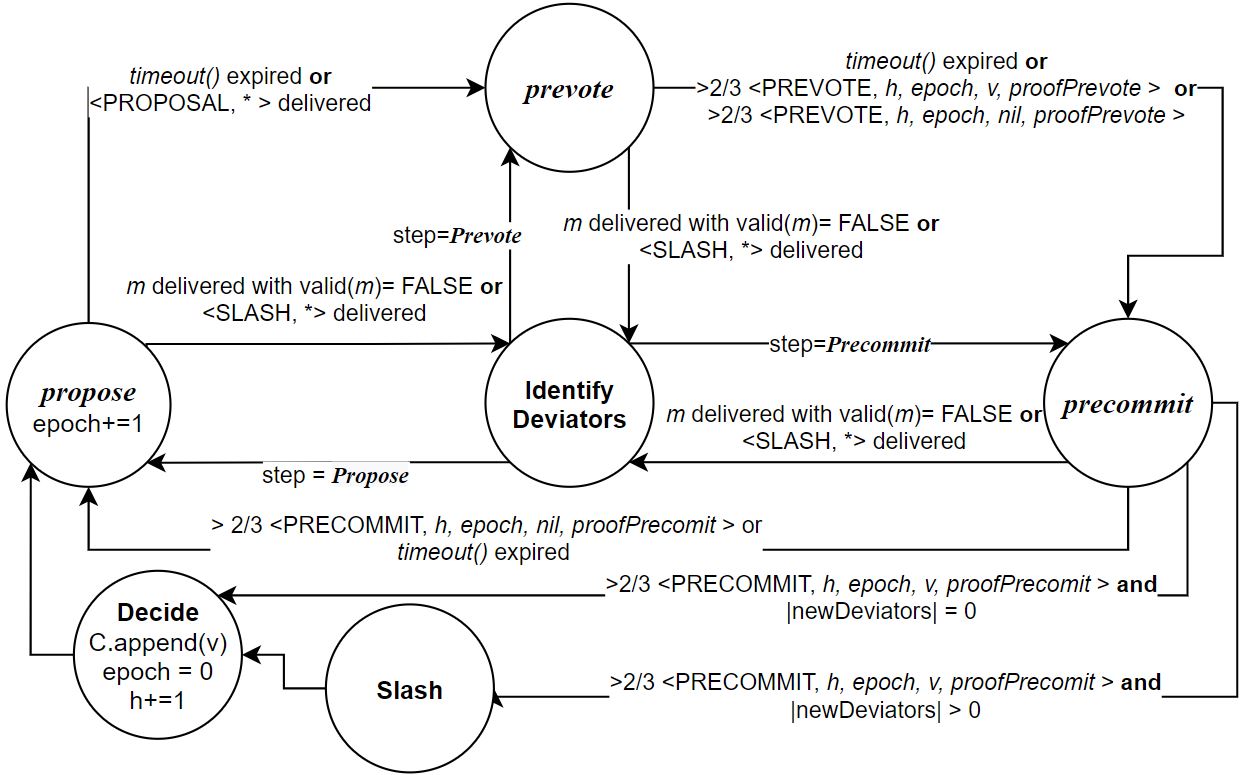}
\caption{A state diagram representation of $\protocolName$}
\end{figure}

%% file: Protocol/ProofOfTransitionFunctionailty.tex
\subsubsection{Proof-of-Transition Functionality}\label{sec:PoT}

In $ \protocolName$, every \\ 
$\Proposal,$ $ \Prevote $ and $ \Precommit$ message must be accompanied by a \textit{proof-of-transition} which evidences the transition to the current step claimed by each player is valid.
These proofs are stored in the local player variables $\prevoteProof $ and  $\precommitProof $, with $\prevoteProof$ also acting as proof for $\Proposal$ messages when a player is selected as proposer.  For example, in line \ref{line:tab:recvAny2/3Precommit}, each $ \Precommit$ message must be accompanied by a $\textit{proof}$ which shows that the respective players were at a protocol step which allowed them to send a $ \Precommit$ message (lines \ref{line:tab:precommit-v}, \ref{line:tab:precommit-v-1} or \ref{line:tab:precommit-nil-onTimeout}). This would be true either if the player received  $\Prevote$ messages correctly satisfying the condition at line \ref{line:nilPrevotes}, both of the conditions at lines \ref{line:tab:recGoodPrevotes}, \ref{line:tab:goodPrevotesAndStepPrevote}, or the condition at \ref{line:tab:onTimeoutPrevote} . As these conditions have specific $\Prevote$ message reception rules, and given there is only one valid $\Precommit$ messages in each case, $\textit{valid}(\textit{proof})$ is true if and only if the message was generated correctly, i.e. by receiving a set of $\Prevote$ messages which would trigger that $\Precommit$ message according to the protocol.

%% file: Protocol/SlashingFunctionality.tex
\subsubsection{Slashing Functionality}\label{sec:Slash}

If any messages/ proofs are not valid in $\protocolName$, players trigger the \textit{Slashing functionality} and send a $\Slash$ message (line \ref{line:deviationIdentified}), which contains a proof that the offending message was indeed invalid (described in detail in Section \ref{sec:PoD}). $\Slash$ messages, along with proofs-of-transition, are a key addition to $\protocolName$ in order to prove $\ourSMR$, as players identified as deviating by a correct player through a $\Slash$ message will eventually be seen by all correct players. When a player is identified as deviating, their proof-of-deviation is added to the local $\deviationProofs$ vector of the observing player. Then when a player is selected as proposer, and $\validValue = \nil$, they add all deviation proofs not already identified in $\playerBlockchain$ to their proposed value (line \ref{line:tab:getValidValue}). After being seen by all correct players, any deviator will eventually be added to a correct player's proposed value and removed from the protocol through the $\removeSlashers$ function (line \ref{line:removeSlashers}). 

The $\removeSlashers$ function takes as input new decided deviators, deletes their stakes (lines \ref{line:destroyDeviatorStake}, \ref{line:adjustStake}), adjusts the remaining player shares to sum to 1 (line \ref{line:adjustShares}), recalibrates the per-block reward to keep the per player reward constant throughout a $\protocolName$ instance (line \ref{line:adjustReward}), and distributes the initial reward $\rewardInitial$ times the initial shares of the deviating players among the remaining players in proportion to their stake (line \ref{line:splitDeviatorPot}).

%% file: Protocol/ProofofDeviation.tex
\subsubsection{Proof-of-Deviation}\label{sec:PoD}

Crucial to the slashing functionality are the \textit{proofs-of-deviation} which can be generated upon the reception of any invalid message. As invalid messages can take various forms, we explicitly define each form of invalid message and how to generate the corresponding proof-of-deviation. Invalid messages in $\protocolName$ can (1) contradict another message from the same sender, (2) propose invalid values, (3) contain an invalid proof-of-deviation or (4) contain an invalid proof-of-transition. 

\begin{remark}
    Any message which does not contain one or more of the deviations outlined in this section is \textit{valid}.
\end{remark}

We do not encode proofs-of-deviation here as their exact implementations are beyond the scope of the paper. However, we describe the maximum amount of information necessary to prove that a player has deviated, which can then be represented in some, possibly condensed form within a $\Slash$ message to prove a player has deviated. In the following we assume a player $\playeri$ performs the corresponding deviation.

\begin{enumerate}
    
    \item Contradictory messages: If a player sends two messages $m$ and $m'$ such that they both contain valid proofs-of-transition, but it is not possible to transition from either of the messages to the other, these messages together constitute a proof of deviation. For example, if there are two messages $m$ and $m'$ from $\playeri$ with the same $\playerHeight$, $\playerEpoch$ and $\playerStep$ tags, or if $\playeri$ proposes a newly generated $\textit{getValue}()$ after sending a $\Precommit$ message in a preceeding epoch at the same height for a different value (which would mean $\validValue \neq \nil$). 
    
    \item Invalid proposed values: As the blockchain value validity predicate is shared by all parties and known a priori, any message from $\playeri$ containing an invalid proposed value $\val$ can be used as a proof of deviation. 
    
    \item Invalid slashing: A slash message is invalid if the accompanying proof-of-deviation is not valid. If the proof-of-deviation is not valid, the corresponding slash message/ proposed value (if it first appears in a proposed value) stands as a proof of deviation.
    
    \item Invalid proof-of-transition: If a message $m$ is received from a player $\playeri$, where $\playeri$ has not attached (a proof of) messages with tag, height, epoch and value variables which validly trigger the logical conditions necessary to send $m$, this constitutes an invalid proof-of-transition. These messages, or lack thereof, constitute a proof-of-deviation. As $\playeri$ signs $m$, the contents of $m$ can be verified, and as such $\playeri$'s attempted proof-of-transition can be proved to belong to $\playeri$ (as the signature must correspond to $m$ and the contained proof-of-transition), and proved to be invalid by all players. Any  message sent by $\playeri$ which does not adhere to one of the protocol-specified $\Broadcast$ formats\footnote{Can be thought of as junk, but also includes attempted communication between players, such as to coordinate collusion.} is considered to contain an invalid proof-of-transition. This is because there is no protocol-specified transition that would create such a message.
 
\end{enumerate}

%% file: AchievesSMR/AchievesSMR.tex
\section{Proving \texorpdfstring{$\protocolName$}{TEXT} achieves \texorpdfstring{$\ourSMR$}{TEXT}}\label{sec:ProtocolProofs}

In this section we prove that $\protocolName$ achieves SMR in the $\ourBAR$ model. To this end, we first prove that it is an SMR protocol when more than $\partialSyncMajority$ of the share is controlled by honest players at all times.

\begin{lemma}\label{lem:noHonestInValues}
    It is not possible to generate a valid deviation proof for an honest player.
\end{lemma}

\begin{proof}
    A valid deviation proof for some player $\playeri$ must identify a message from $\playeri$ that is invalid according to one of the methods listed in Section \ref{sec:PoD}, which covers all possible message deviations. By definition, honest players follow all protocol rules and only send valid messages. We also know that honest player messages cannot be forged under our threat model assumption regarding unforgeable signatures from Section \ref{sec:ThreatModel}. Furthermore, as $\adversary$ is static, every message signed by a currently honest player must have been generated honestly (by that same player) at some point in the protocol. Therefore, all honest player messages are valid, and as such, no valid proof-of-deviation described in Section \ref{sec:PoD} can be generated for honest players. 
\end{proof}

\begin{lemma}\label{lem:AchieveSMR}
    $\protocolName$ achieves SMR when players controlling more than $\partialSyncMajority$ of the stake are honest.
\end{lemma}

\begin{proof}
    (\textit{Sketch}) We outline a proof demonstrating  that proposed values (line \ref{line:tab:getValidValue}) satisfy safety and liveness in $\protocolName$. An initialisation of $\protocolName$ is equivalent to a standard Tendermint initialisation, in addition to the proof-of-transition and slash functionalities as described in Sections \ref{sec:PoT} and \ref{sec:Slash} respectively.
    For deciding on a value at a particular height $\blockchainHeight>1$, voting share is described by the value decided at height $\blockchainHeight-1$, or the current proposed value (line \ref{line:tab:getValidValue}) if it is valid and contains newly identified deviators. From Lemma \ref{lem:noHonestInValues}, we know no valid proof of deviation can be generated for an honest player. Therefore, honest players can never be included as prospective deviators in valid proposed values in line \ref{line:tab:getValidValue}. This ensures that any valid value proposed will maintain honest voting share of more than $\partialSyncMajority$. 
    
    Proof-of-transition values are simply additional pieces of information attached to standard Tendermint messages. Identically to Tendermint, $\protocolName$ does not consider invalid messages for any of the steps needed to decide on a block (lines \ref{line:tab:accept-proposal-2}, \ref{line:tab:recGoodPrevotes},  \ref{line:tab:onDecideRule}, \ref{line:tab:acceptProposal}). Consider an epoch during synchrony, with timeouts larger than the message delivery delta $\Delta$ for all honest players (necessary to ensure liveness, as in Tendermint) and an honest proposer. This epoch occurs eventually at every height (if no decision has been reached in earlier epochs) as the network communication model and proposer rotation are identical to Tendermint, and the timeout function is increasing and unbounded in the number of epochs. As more than $\partialSyncMajority$ of the voting share is controlled by honest players at all times, honest players will decide on the proposed value during that epoch. This holds for any height $\blockchainHeight>1$, and the safety and liveness of $\protocolName$ follows.
\end{proof}

With $\protocolName$ as an SMR protocol under an honest majority, we now need to prove $\protocolName$ achieves $\ourSMR$. To do this, we will prove that $\protocolName$ is SINCE and fair in the $\ourBAR$ model, and apply Theorem \ref{thm:fair+SICIE=Necessary+Sufficient}. To prove SINCE, we first need some results that bound the reward a player can achieve for deciding on a value. As each decided value requires an accompanying $> \partialSyncMajority$ $\Precommit$ messages to be valid, the maximum amount of values a player can attempt to decide on at once is 1, as the proceeding value will need to point to the decided value and the value's $> \partialSyncMajority$ $\Precommit$ messages. By bounding the reward a player can get for deciding on values and deviators (the only rewarding actions) sufficiently low, we are able to prove that this reward is negligible compared to the potential punishment for being caught, thus preventing rational players from sending invalid messages.

\begin{lemma} \label{lem:boundedDeviatorPot}
    In any instance of $\protocolName$, $\playeri$ receives less than $\sharei^{\initialisationHeight} \rewardInitial$ in total for identifying deviators from the \textit{adjustForSlashing} function.
\end{lemma}

\begin{proof}
    We can see that the rewards for deciding on new deviators at height $\blockchainHeight$ are distributed at line \ref{line:splitDeviatorPot}. We will prove that the sum of the rewards distributed by calling line \ref{line:splitDeviatorPot} throughout an instance of $\protocolName$ are less than $\sharei^{\initialisationHeight} \rewardInitial$.

    Let there be a set of players $\textit{newDevs}^\blockchainHeight$ controlling $\slashedShare$ at height $\blockchainHeight-1$ identified as deviating for the first time in the value at height $\blockchainHeight$. 
    In the \textit{adjustForSlashing} function, this results in $\playeri$'s share being updated according to line \ref{line:adjustShares}, which implies $\shareHi =\frac{\sharei^{\blockchainHeight-1}}{1-\slashedShare}$. Furthermore, letting $\reward^{\blockchainHeight-1}$ be the total reward after deciding on a value at height $\blockchainHeight-1$, the new total reward for height $\blockchainHeight$ is $\reward^\blockchainHeight= (1-\slashedShare) \reward^{\blockchainHeight-1}  $ (line \ref{line:adjustReward}), while the total stake before distributing rewards for height $\blockchainHeight$ is adjusted to $(1-\slashedShare) \stake^{\blockchainHeight-1} $ (line \ref{line:adjustStake}), preserving the total stake of non-deviators. We then have to add the slash bonus of $\sum_{j \in \textit{newDevs}^\blockchainHeight} $ $\textit{Genesis.shares}()[j] \reward^\blockchainHeight$  (line \ref{line:splitDeviatorPot}).

    First observe that:  
    \begin{equation}
    \begin{split}
         \sharei^{\blockchainHeight} \reward^{\blockchainHeight} &= \frac{\sharei^{\blockchainHeight-1}}{1-\slashedShare} (1-\slashedShare) \reward^{\blockchainHeight-1}  \\ 
        &= \sharei^{\blockchainHeight-1}  \reward^{\blockchainHeight-1}.
    \end{split}
    \end{equation}
    
    Secondly, notice that when no new deviators are identified, and \textit{adjustForSlashing} is not called, both $\sharei$ and $\reward_i$ are unchanged from the previous height, as they are only adjusted in \textit{adjustForSlashing}. This means: 
    \begin{equation} \label{eq:ConstantReward}
        \sharei^{\blockchainHeight} \reward^{\blockchainHeight}
        = \sharei^\initialisationHeight \reward^{\initialisationHeight}
        =\sharei^\initialisationHeight \rewardInitial, \ \forall \ \blockchainHeight \geq 1.
    \end{equation}
    
    We know for a set of new deviators $\textit{newDevs}^\blockchainHeight$ at height $\blockchainHeight$, $\playeri$ receives $\sharei^{\blockchainHeight} \reward^{\blockchainHeight} $ $ \sum_{j \in {\textit{newDevs}^\blockchainHeight}} $ $ \textit{Genesis.shares}()[j] $(line \ref{line:splitDeviatorPot}). 
    Furthermore, from Equation \ref{eq:ConstantReward}, we have that  $\sharei^{\blockchainHeight} \reward^{\blockchainHeight}= \sharei^\initialisationHeight $ $ \rewardInitial$ for all $\blockchainHeight \geq 1$. 
    This implies $\playeri$ receives an identifying deviator bonus from \textit{adjustForSlashing} of $\sharei^\initialisationHeight   \rewardInitial \cdot $ $\sum_{j \in {\textit{newDevs}^\blockchainHeight}}$ $ \textit{Genesis.shares}()[j]$ at height $\blockchainHeight$. 
    Summing over all heights up to and including $\blockchainHeight$ gives a total reward of $\sharei^\initialisationHeight $ \\ $  \rewardInitial \cdot \sum^\blockchainHeight_{k=1}  $ $\big( \sum_{j \in {\textit{newDevs}^k}} $ $\textit{Genesis.shares}()[j] \big)$ for identifying deviators through \textit{adjustForSlashing}. 
    
    As $\cup_{1 \leq k\leq \blockchainHeight} {\textit{newDevs}^k} \subset \{1, ..., \numPlayers\}$ for all $\blockchainHeight > 1$, it must be that $\sum^\blockchainHeight_{k=1} \big( \sum_{j \in {\textit{newDevs}^k}} $ $\textit{Genesis.shares}()[j] \big) < 1$ for all $\blockchainHeight > 1$.
    This means the total reward for identifying deviators in $\protocolName$ through the \textit{adjustForSlashing} function is less than $\sharei^\initialisationHeight  \rewardInitial$, as required. 
\end{proof}

\begin{lemma}\label{lem:constantReward}
    In addition to any rewards from the \textit{adjustForSlashing} function, $\playeri$ receives $\sharei^{\initialisationHeight} \rewardInitial$ for every decided value in $\protocolName$.
\end{lemma}

\begin{proof}
    The only reward received by $\playeri$ not in \textit{adjustForSlashing} is distributed at line \ref{line:addReward}. Letting $\reward^{\blockchainHeight}$ be the total reward distributed at line \ref{line:addReward} for height $\blockchainHeight$, $\playeri$ receives $\sharei^{\blockchainHeight} \reward^{\blockchainHeight}$. We have already seen in Equation \ref{eq:ConstantReward} that $\sharei^{\blockchainHeight} \reward^{\blockchainHeight}= \sharei^\initialisationHeight  \rewardInitial$, for all $\blockchainHeight \geq 1$, which is the required result.

\end{proof}

\begin{remark}
    In $\protocolName$, share increases are counteracted by reward decreases to keep per-decision rewards constant (Lemma \ref{lem:constantReward}). This avoids a common, critical, mistake in incentive compatible reward mechanisms where early share increases permanently increase the size of per-decision rewards a player receives.
\end{remark}

\begin{lemma}\label{lem:SINCE}
    $\protocolName$ is SINCE in the ByRa model.
\end{lemma}

\begin{proof}
    To prove SINCE in the ByRa model, we require that every protocol action strictly dominates all other possible actions in expectation for rational players assuming all other players are rational. We do this by proving the following:
    
    \begin{enumerate}
        \item Rational players do not send invalid messages.
        \item Rational players send valid messages when possible.
        \item Rational players obey a timeout function which is increasing and unbounded in epochs at every height. 
    \end{enumerate}
    
    Firstly, consider invalid protocol messages. As an invalid message takes one of the forms described in Section \ref{sec:PoD}, it can eventually be identified by all players and the offending player stake destroyed through the Slashing functionality (Section \ref{sec:Slash}). As identifying deviations of other players is strictly increasing in stake (line \ref{line:splitDeviatorPot}) and does not affect proceeding rewards due to Lemma \ref{lem:constantReward}, all rational players prefer to eventually identify valid deviations than not identify valid deviations. As stake is only meaningful with respect to a valid blockchain, $\playeri$ must construct $\playerBlockchain$ sequentially in it's height. Therefore, for $\playerHeight$ the height of $\playerBlockchain$, $\playeri$'s messages can only refer to a value at a height less than or equal to $\playerHeight+1$.
    
    Combining Lemmas \ref{lem:boundedDeviatorPot} and \ref{lem:constantReward}, for any height $\blockchainHeight>1$, the maximum additional reward achievable by sending an invalid message up to that height is less than $2 \sharei^{\initialisationHeight} \rewardInitial$. This is because by Lemma \ref{lem:boundedDeviatorPot}, the additional reward for identifying deviators is less than $ \sharei^{\initialisationHeight} \rewardInitial$, and by Lemma \ref{lem:constantReward}, the per-decided value reward excluding any reward for identifying deviators is $ \sharei^{\initialisationHeight} \rewardInitial$, a constant. Therefore, attempting to decide on a value and/ or deviators (the only ways to be rewarded in $\protocolName$) with an invalid message results in a payoff of less than $ 2 \sharei^{\initialisationHeight} \rewardInitial$. Due to Lemma \ref{lem:constantReward}, the rewards for proceeding value-decisions remain constant, while all rewards for identifying deviators must sum up to less than $\sharei^{\initialisationHeight} \rewardInitial$ from Lemma \ref{lem:boundedDeviatorPot}. 
    Given proofs-of-deviation can be provided at any time and all rational players will send $\Slash$ messages when possible, the cost of sending an invalid message, full destruction of stake (line \ref{line:destroyDeviatorStake}) and effective removal from the protocol, dominates these once-off and bounded potential rewards for sending an invalid message. As such, no rational player will send an invalid message. 
    
    Given no rational player will send an invalid message, we now need to check that rational players will send messages when valid messages can be sent, as per the protocol. The alternative is not sending messages. Given the arbitrary scheduling of message delivery in any distributed network where other players have unknown timeouts, and the positive reward for deciding on a block, sending messages strictly increases the expected rate of messages received by all other players. This in turn strictly increases the expected rate of player progression through the protocol, as progression can only occur when proofs can be generated. This strictly increases the expected number of blocks, and rewards, added to the blockchain. 
    
    Lastly, we must ensure that rational players obey a timeout function which tends to infinity in the number of epochs at each height. To do this we first show that rational players obey some non-zero timeout, and then that this timeout is increasing and unbounded in number of epochs.
    
    If a rational player does not wait for messages to be delivered, they will never be able to contribute to prevotes for valid values unless they are a proposer. After entering a new epoch they will call line \ref{line:timeoutPropose}, immediately followed by line \ref{line:tab:onTimeoutPropose}, sending a \textit{nil} prevote. Moreover, given they send a \textit{nil} prevote and advance to the $\prevote$ step, they will also send a \textit{nil} precommit (line \ref{line:tab:onTimeoutPrevote}) as when they receive more than $ \partialSyncMajority$ prevotes it includes their own $\nil$ prevote, triggering line \ref{line:tab:recvAny2/3Prevote} before it is possible to receive more than $ \partialSyncMajority$ prevotes for a valid value. By the same argumentation, they will never be able to decide on a value in the epoch it is proposed as they will first receive more than $ \partialSyncMajority$ precommits for inconsistent values given their \textit{nil} precommit message, triggering line \ref{line:tab:recvAny2/3Precommit} and then immediately line \ref{line:tab:onTimeoutPrecommit}, preventing a decision. Compare this to obeying some timeout for messages to be delivered. Waiting for some number of rounds strictly increases the probability of receiving valid proposed values, and sending a prevote for a valid value. This subsequently increases the probability of all players sending valid precommits. By further obeying a timeout for precommits it increases the probability of receiving the quorum of precommits needed to decide on a value. Therefore, rational players prefer to wait some number of rounds for messages to be delivered. 
    
    Now we must ensure rational players do not wait indefinitely for messages. 
    Recall that in Tenderstake, rational players are modelled as assuming for some unknown but fixed $\Delta$, they are in $\Delta$-synchrony with some subset of players.
    At any round $\round$, in order to calculate expected utility for some future round, $\playeri$ will have a private distribution of expected message delivery times from other players in synchrony\footnote{The distribution of expected message delivery times will be a function of some starting estimate at initialisation (perhaps based on a \textit{Genesis} suggested value, as in \cite{LatestGossipTendermint}), and the observed responsiveness of all other players up until round $\round$.}, and thus an expected number of decisions up until that future round. Let $\timeout^\round_i$ be such that according to $\playeri$'s private information, messages taking longer than $\timeout^\round_i$ are sent by players out of synchrony with $\playeri$ with statistical significance\footnote{This statistical significance can be with respect to a function $\negligible$, although rational players may perceive a higher utility by choosing a weaker significance level. This optimisation is unnecessary for the proof.}. 
    
    If the subset of players in synchrony with $\playeri$, including $\playeri$, do not control more than $\partialSyncMajority $ of the total stake, $\playeri$ is indifferent to timing out, as no decision is possible. 
    Otherwise, consider the subset of players in synchrony with $\playeri$, including $\playeri$, controlling more than $\partialSyncMajority $ of the total stake. As messages sent by players out of synchrony with $\playeri$ take arbitrarily long to deliver, the number of decisions that can be made by using a timeout of $\timeout^\round_i$ and transitioning to a proposer in synchrony with $\playeri$ is arbitrarily large. Furthermore, as $\playeri$ is unaware of how many players are in synchrony with $\playeri$, the probability of that subset controlling more than $\partialSyncMajority $ of the stake will be positive. This implies $\playeri$ has positive expectancy to obey such a timeout $\timeout^\round_i$. Therefore, rational players obey some timeout, and will not wait indefinitely for messages. 
    
    We finally need to show that for a rational $\playeri$ at any given height, $\playeri$ will follow increasing, unbounded timeouts in the number of epochs at every height. 
    For a maximum message delivery time of $\Delta$ rounds during synchrony, if $\playeri$ follows a timeout of $\timeouti< \Delta $, $\playeri$ is not necessarily able to contribute to deciding on a value. Assume players controlling more than $\partialSyncMajority$ of the stake are in $\Delta$-synchrony (if this is not the case, no information can be gained) and no decision has been made for some number of epochs. As players behave honestly in all non-timeout actions (points 1 and 2), the only variable which can affect the probability of deciding for this height is $\timeouti$. Assume for all rational players there is a value $\timeout_{\textit{max}}>0$, such that they choose timeouts less than $ \timeout_{\textit{max}}$ for all epochs. 

    If $\timeout_{\textit{max}}< \Delta$, it is possible that players may always timeout, sending \textit{nil} messages and not contributing to decisions. Given there has been $\epoch$ epochs of not deciding on a value, and all other actions are being followed (which we have shown to be the case), it must be that $P(\timeout_{\textit{max}}< \Delta | \epoch \rightarrow \infty) \rightarrow 1$. This implies choosing a timeout up to and including $\timeout_{\textit{max}}$ after sufficiently many epochs of no decision results in decision with $\negligible$ probability for proceeding epochs. Therefore, rational players will eventually only follow timeouts greater than $\timeout_{\textit{max}}$ if no value has been decided, for any value of $\timeout_{\textit{max}}$. 
    
    This is sufficient to say rational players follow increasing, unbounded timeouts, and as such, the recommended protocol.
\end{proof}

\begin{lemma}\label{lem:Fair}
    $\protocolName$ is fair in the ByRa model.
\end{lemma}

\begin{proof}
    As all rational players follow the protocol, and $\share_\adversary^{\initialisationHeight}< \partialSyncAdvLimit$, no rational player decides on another rational player as deviating. Therefore, the share of stake controlled by rational players is only increasing (line \ref{line:adjustShares}), meaning the adversary's share is only decreasing, upperbounded by their starting share. This implies  $\share_\adversary^{\blockchainHeight} \leq \share_\adversary^{\initialisationHeight}$ for all $\blockchainHeight \geq \initialisationHeight$, which is precisely the definition of a fair protocol. 
\end{proof}

\begin{theorem}
    $\protocolName$ achieves ByRa SMR.
\end{theorem}

\begin{proof}
    Follows by applying Lemma \ref{lem:SINCE} and Lemma \ref{lem:Fair} to Theorem \ref{thm:fair+SICIE=Necessary+Sufficient}.
\end{proof}

%% file: Conclusion.tex
\section{Conclusion}\label{sec:conclusion}

We provide a game-theoretic framework for analysing SMR protocols. Although many previous attempts have been made, we are, to the best of our knowledge, the first to formally treat SMR protocols as games involving only rational and adversarial players. We detail the $\ourBAR$ model for player characterisation in SMR protocols, an update to the legacy BAR model, removing the dependency on altruistic players in an era of unprecedented market capitalisation of tokenised SMR protocols. We demonstrate that the properties of strong incentive compatibility in expectation and fairness as described in this paper, are both necessary, and together sufficient to achieve SMR in the $\ourBAR$ model. We then provide the $\protocolName$ protocol as an example of a protocol that achieves $\ourSMR$, which is of independent interest both as a strong incentive compatible in expectation and fair protocol in the $\ourBAR$ model, but also as a yardstick for addressing the shortcomings of current protocol guarantees in the $\ourBAR$ model. The proof techniques we use provide several methodologies with which SMR protocols can be analysed in this new game-theoretic framework.
The improvements we make to the Tendermint protocol as described in Section \ref{sec:Protocol} have immediate practical implications given the current industrial deployment of Tendermint-style protocols, such as in Cosmos\footnote{Cosmos. \url{https://cosmos.network/} Accessed: 25/05/2021}. The application of our framework to all future SMR protocol analysis and development serves as critical future work. Another important consideration for future work is that of the $\ourBAR$ model under an adaptive adversary as stated in Remark \ref{rem:StaticAdv}.

%% file: Tenderstake.bbl

\begin{thebibliography}{36}


\ifx \showCODEN    \undefined \def \showCODEN     #1{\unskip}     \fi
\ifx \showDOI      \undefined \def \showDOI       #1{#1}\fi
\ifx \showISBNx    \undefined \def \showISBNx     #1{\unskip}     \fi
\ifx \showISBNxiii \undefined \def \showISBNxiii  #1{\unskip}     \fi
\ifx \showISSN     \undefined \def \showISSN      #1{\unskip}     \fi
\ifx \showLCCN     \undefined \def \showLCCN      #1{\unskip}     \fi
\ifx \shownote     \undefined \def \shownote      #1{#1}          \fi
\ifx \showarticletitle \undefined \def \showarticletitle #1{#1}   \fi
\ifx \showURL      \undefined \def \showURL       {\relax}        \fi
\providecommand\bibfield[2]{#2}
\providecommand\bibinfo[2]{#2}
\providecommand\natexlab[1]{#1}
\providecommand\showeprint[2][]{arXiv:#2}

\bibitem[\protect\citeauthoryear{Abraham, Devadas, Dolev, Nayak, and
  Ren}{Abraham et~al\mbox{.}}{2017}]%
        {SyncByzCon}
\bibfield{author}{\bibinfo{person}{Ittai Abraham}, \bibinfo{person}{Srinivas
  Devadas}, \bibinfo{person}{Danny Dolev}, \bibinfo{person}{Kartik Nayak},
  {and} \bibinfo{person}{Ling Ren}.} \bibinfo{year}{2017}\natexlab{}.
\newblock \bibinfo{title}{Efficient Synchronous Byzantine Consensus}.
\newblock \bibinfo{howpublished}{\url{https://eprint.iacr.org/2017/307}}.
\newblock
\newblock
\shownote{Retrieved: 18/05/2021.}


\bibitem[\protect\citeauthoryear{Abraham, Malkhi, Nayak, Ren, and
  Spiegelman}{Abraham et~al\mbox{.}}{2018}]%
        {Solida}
\bibfield{author}{\bibinfo{person}{Ittai Abraham}, \bibinfo{person}{Dahlia
  Malkhi}, \bibinfo{person}{Kartik Nayak}, \bibinfo{person}{Ling Ren}, {and}
  \bibinfo{person}{Alexander Spiegelman}.} \bibinfo{year}{2018}\natexlab{}.
\newblock \showarticletitle{{Solida: A Blockchain Protocol Based on
  Reconfigurable Byzantine Consensus}}. In \bibinfo{booktitle}{\emph{21st
  International Conference on Principles of Distributed Systems (OPODIS 2017)}}
  \emph{(\bibinfo{series}{Leibniz International Proceedings in Informatics
  (LIPIcs)}, Vol.~\bibinfo{volume}{95})},
  \bibfield{editor}{\bibinfo{person}{James Aspnes}, \bibinfo{person}{Alysson
  Bessani}, \bibinfo{person}{Pascal Felber}, {and} \bibinfo{person}{Jo{\~a}o
  Leit{\~a}o}} (Eds.). \bibinfo{publisher}{Schloss Dagstuhl--Leibniz-Zentrum
  fuer Informatik}, \bibinfo{address}{Dagstuhl, Germany},
  \bibinfo{pages}{25:1--25:19}.
\newblock
\showISBNx{978-3-95977-061-3}
\showISSN{1868-8969}
\urldef\tempurl%
\url{https://doi.org/10.4230/LIPIcs.OPODIS.2017.25}
\showDOI{\tempurl}


\bibitem[\protect\citeauthoryear{Aiyer, Alvisi, Clement, Dahlin, Martin, and
  Porth}{Aiyer et~al\mbox{.}}{2005}]%
        {BAR-FT}
\bibfield{author}{\bibinfo{person}{Amitanand~S. Aiyer},
  \bibinfo{person}{Lorenzo Alvisi}, \bibinfo{person}{Allen Clement},
  \bibinfo{person}{Mike Dahlin}, \bibinfo{person}{Jean-Philippe Martin}, {and}
  \bibinfo{person}{Carl Porth}.} \bibinfo{year}{2005}\natexlab{}.
\newblock \showarticletitle{BAR Fault Tolerance for Cooperative Services}.
\newblock \bibinfo{journal}{\emph{SIGOPS Oper. Syst. Rev.}}
  \bibinfo{volume}{39}, \bibinfo{number}{5} (\bibinfo{date}{Oct.}
  \bibinfo{year}{2005}), \bibinfo{pages}{45–58}.
\newblock
\showISSN{0163-5980}
\urldef\tempurl%
\url{https://doi.org/10.1145/1095809.1095816}
\showDOI{\tempurl}


\bibitem[\protect\citeauthoryear{Alsabah and Capponi}{Alsabah and
  Capponi}{2020}]%
        {BitcoinRnDArmsRace}
\bibfield{author}{\bibinfo{person}{Humoud Alsabah} {and}
  \bibinfo{person}{Agostino Capponi}.} \bibinfo{year}{2020}\natexlab{}.
\newblock \bibinfo{title}{Pitfalls of Bitcoin’s Proof-of-Work: R\&D Arms Race
  and Mining Centralization}.
\newblock \bibinfo{howpublished}{\url{https://ssrn.com/abstract=3273982}}.
\newblock
\newblock
\shownote{Retrieved: 18/05/2021.}


\bibitem[\protect\citeauthoryear{Amoussou-Guenou, Biais, Potop-Butucaru, and
  Tucci-Piergiovanni}{Amoussou-Guenou et~al\mbox{.}}{2020}]%
        {RationalsvsByzantinesConsensusBasedBlockchains}
\bibfield{author}{\bibinfo{person}{Yackolley Amoussou-Guenou},
  \bibinfo{person}{Bruno Biais}, \bibinfo{person}{Maria Potop-Butucaru}, {and}
  \bibinfo{person}{Sara Tucci-Piergiovanni}.} \bibinfo{year}{2020}\natexlab{}.
\newblock \showarticletitle{Rational vs Byzantine Players in Consensus-Based
  Blockchains}. In \bibinfo{booktitle}{\emph{Proceedings of the 19th
  International Conference on Autonomous Agents and MultiAgent Systems}}
  (Auckland, New Zealand) \emph{(\bibinfo{series}{AAMAS '20})}.
  \bibinfo{publisher}{International Foundation for Autonomous Agents and
  Multiagent Systems}, \bibinfo{address}{Richland, SC},
  \bibinfo{pages}{43–51}.
\newblock
\showISBNx{9781450375184}


\bibitem[\protect\citeauthoryear{Amoussou-Guenou, Biais, Potop-Butucaru, and
  Tucci-Piergiovanni}{Amoussou-Guenou et~al\mbox{.}}{2021a}]%
        {RationalBehaviorCommitteeBasedBlockchains}
\bibfield{author}{\bibinfo{person}{Yackolley Amoussou-Guenou},
  \bibinfo{person}{Bruno Biais}, \bibinfo{person}{Maria Potop-Butucaru}, {and}
  \bibinfo{person}{Sara Tucci-Piergiovanni}.} \bibinfo{year}{2021}\natexlab{a}.
\newblock \showarticletitle{{Rational Behaviors in Committee-Based
  Blockchains}}. In \bibinfo{booktitle}{\emph{24th International Conference on
  Principles of Distributed Systems (OPODIS 2020)}}
  \emph{(\bibinfo{series}{Leibniz International Proceedings in Informatics
  (LIPIcs)}, Vol.~\bibinfo{volume}{184})},
  \bibfield{editor}{\bibinfo{person}{Quentin Bramas}, \bibinfo{person}{Rotem
  Oshman}, {and} \bibinfo{person}{Paolo Romano}} (Eds.).
  \bibinfo{publisher}{Schloss Dagstuhl--Leibniz-Zentrum f{\"u}r Informatik},
  \bibinfo{address}{Dagstuhl, Germany}, \bibinfo{pages}{12:1--12:16}.
\newblock
\showISBNx{978-3-95977-176-4}
\showISSN{1868-8969}
\urldef\tempurl%
\url{https://doi.org/10.4230/LIPIcs.OPODIS.2020.12}
\showDOI{\tempurl}


\bibitem[\protect\citeauthoryear{Amoussou-Guenou, Pozzo, Potop-Butucaru, and
  Tucci-Piergiovanni}{Amoussou-Guenou et~al\mbox{.}}{2018}]%
        {FairnessTendermint2018}
\bibfield{author}{\bibinfo{person}{Yackolley Amoussou-Guenou},
  \bibinfo{person}{Antonella~Del Pozzo}, \bibinfo{person}{Maria
  Potop-Butucaru}, {and} \bibinfo{person}{Sara Tucci-Piergiovanni}.}
  \bibinfo{year}{2018}\natexlab{}.
\newblock \bibinfo{title}{Correctness and Fairness of Tendermint-core
  Blockchains}.
\newblock \bibinfo{howpublished}{\url{https://arxiv.org/pdf/1805.08429 }}.
\newblock
\showeprint{1805.08429}
\newblock
\shownote{Retrieved: 18/05/2021.}


\bibitem[\protect\citeauthoryear{Amoussou-Guenou, Pozzo, Potop-Butucaru, and
  Tucci-Piergiovanni}{Amoussou-Guenou et~al\mbox{.}}{2021b}]%
        {FairnessTendermint2019}
\bibfield{author}{\bibinfo{person}{Yackolley Amoussou-Guenou},
  \bibinfo{person}{Antonella~Del Pozzo}, \bibinfo{person}{Maria
  Potop-Butucaru}, {and} \bibinfo{person}{Sara Tucci-Piergiovanni}.}
  \bibinfo{year}{2021}\natexlab{b}.
\newblock \showarticletitle{{On Fairness in Committee-Based Blockchains}}. In
  \bibinfo{booktitle}{\emph{2nd International Conference on Blockchain
  Economics, Security and Protocols (Tokenomics 2020)}}
  \emph{(\bibinfo{series}{Open Access Series in Informatics (OASIcs)},
  Vol.~\bibinfo{volume}{82})}, \bibfield{editor}{\bibinfo{person}{Emmanuelle
  Anceaume}, \bibinfo{person}{Christophe Bisi\`{e}re},
  \bibinfo{person}{Matthieu Bouvard}, \bibinfo{person}{Quentin Bramas}, {and}
  \bibinfo{person}{Catherine Casamatta}} (Eds.). \bibinfo{publisher}{Schloss
  Dagstuhl--Leibniz-Zentrum f{\"u}r Informatik}, \bibinfo{address}{Dagstuhl,
  Germany}, \bibinfo{pages}{4:1--4:15}.
\newblock
\showISBNx{978-3-95977-157-3}
\showISSN{2190-6807}
\urldef\tempurl%
\url{https://doi.org/10.4230/OASIcs.Tokenomics.2020.4}
\showDOI{\tempurl}


\bibitem[\protect\citeauthoryear{Arnosti and {Matthew Weinberg}}{Arnosti and
  {Matthew Weinberg}}{2019}]%
        {BitcoinANaturalOligopoly}
\bibfield{author}{\bibinfo{person}{Nick Arnosti} {and} \bibinfo{person}{S.
  {Matthew Weinberg}}.} \bibinfo{year}{2019}\natexlab{}.
\newblock \showarticletitle{Bitcoin: A natural oligopoly}. In
  \bibinfo{booktitle}{\emph{10th Innovations in Theoretical Computer Science,
  ITCS 2019}} \emph{(\bibinfo{series}{Leibniz International Proceedings in
  Informatics, LIPIcs})}, \bibfield{editor}{\bibinfo{person}{Avrim Blum}}
  (Ed.). \bibinfo{publisher}{Schloss Dagstuhl- Leibniz-Zentrum fur Informatik
  GmbH, Dagstuhl Publishing}, \bibinfo{address}{Germany}.
\newblock
\urldef\tempurl%
\url{https://doi.org/10.4230/LIPIcs.ITCS.2019.5}
\showDOI{\tempurl}
\newblock
\shownote{Funding Information: Supported by NSF CCF-1717899.; 10th Innovations
  in Theoretical Computer Science, ITCS 2019 ; Conference date: 10-01-2019
  Through 12-01-2019.}


\bibitem[\protect\citeauthoryear{Azouvi and Hicks}{Azouvi and Hicks}{2020}]%
        {SoKToolsGameTheoryCryptocurrencies}
\bibfield{author}{\bibinfo{person}{Sarah Azouvi} {and}
  \bibinfo{person}{Alexander Hicks}.} \bibinfo{year}{2020}\natexlab{}.
\newblock \bibinfo{title}{SoK: Tools for Game Theoretic Models of Security for
  Cryptocurrencies}.
\newblock \bibinfo{howpublished}{\url{https://arxiv.org/abs/1905.08595}}.
\newblock
\showeprint{1905.08595}
\newblock
\shownote{Retrieved: 19/05/2021.}


\bibitem[\protect\citeauthoryear{Bano, Sonnino, Al-Bassam, Azouvi, McCorry,
  Meiklejohn, and Danezis}{Bano et~al\mbox{.}}{2019}]%
        {SoKConsensus}
\bibfield{author}{\bibinfo{person}{Shehar Bano}, \bibinfo{person}{Alberto
  Sonnino}, \bibinfo{person}{Mustafa Al-Bassam}, \bibinfo{person}{Sarah
  Azouvi}, \bibinfo{person}{Patrick McCorry}, \bibinfo{person}{Sarah
  Meiklejohn}, {and} \bibinfo{person}{George Danezis}.}
  \bibinfo{year}{2019}\natexlab{}.
\newblock \showarticletitle{SoK: Consensus in the Age of Blockchains}. In
  \bibinfo{booktitle}{\emph{Proceedings of the 1st ACM Conference on Advances
  in Financial Technologies}} (Zurich, Switzerland) \emph{(\bibinfo{series}{AFT
  ’19})}. \bibinfo{publisher}{Association for Computing Machinery},
  \bibinfo{address}{New York, NY, USA}, \bibinfo{pages}{183--198}.
\newblock
\showISBNx{9781450367325}
\urldef\tempurl%
\url{https://doi.org/10.1145/3318041.3355458}
\showDOI{\tempurl}


\bibitem[\protect\citeauthoryear{Biais, Bisière, Bouvard, and Casamatta}{Biais
  et~al\mbox{.}}{2017}]%
        {BlockchainFolkTheorem}
\bibfield{author}{\bibinfo{person}{Bruno Biais}, \bibinfo{person}{Christophe
  Bisière}, \bibinfo{person}{Matthieu Bouvard}, {and}
  \bibinfo{person}{Catherine Casamatta}.} \bibinfo{year}{2017}\natexlab{}.
\newblock \bibinfo{booktitle}{\emph{The blockchain folk theorem}}.
\newblock \bibinfo{type}{IDEI Working Papers} 873.
  \bibinfo{institution}{Institut d'Économie Industrielle (IDEI), Toulouse}.
\newblock
\newblock
\shownote{Retrieved: 19/05/2021.}


\bibitem[\protect\citeauthoryear{Birmpas, Koutsoupias, Lazos, and
  Marmolejo-Cossío}{Birmpas et~al\mbox{.}}{2020}]%
        {FairnessDAGs}
\bibfield{author}{\bibinfo{person}{Georgios Birmpas}, \bibinfo{person}{Elias
  Koutsoupias}, \bibinfo{person}{Philip Lazos}, {and}
  \bibinfo{person}{Francisco~J. Marmolejo-Cossío}.}
  \bibinfo{year}{2020}\natexlab{}.
\newblock \showarticletitle{Fairness and Efficiency in DAG-Based
  Cryptocurrencies}. In \bibinfo{booktitle}{\emph{Financial Cryptography and
  Data Security}}. \bibinfo{publisher}{Springer International Publishing},
  \bibinfo{address}{Cham}, \bibinfo{pages}{79--96}.
\newblock
\showISBNx{978-3-030-51280-4}


\bibitem[\protect\citeauthoryear{Buchman, Kwon, and Milosevic}{Buchman
  et~al\mbox{.}}{2019}]%
        {LatestGossipTendermint}
\bibfield{author}{\bibinfo{person}{Ethan Buchman}, \bibinfo{person}{Jae Kwon},
  {and} \bibinfo{person}{Zarko Milosevic}.} \bibinfo{year}{2019}\natexlab{}.
\newblock \bibinfo{title}{The latest gossip on BFT consensus}.
\newblock \bibinfo{howpublished}{\url{https://arxiv.org/abs/1807.049385}}.
\newblock
\showeprint[arxiv]{1807.04938}
\newblock
\shownote{Retrieved: 21/05/2021.}


\bibitem[\protect\citeauthoryear{Budish}{Budish}{2018}]%
        {EconomicLimitsBitcoinBlockchain}
\bibfield{author}{\bibinfo{person}{Eric Budish}.}
  \bibinfo{year}{2018}\natexlab{}.
\newblock \bibinfo{booktitle}{\emph{The Economic Limits of Bitcoin and the
  Blockchain}}.
\newblock \bibinfo{type}{Working Paper} 24717. \bibinfo{institution}{National
  Bureau of Economic Research}.
\newblock
\urldef\tempurl%
\url{https://doi.org/10.3386/w24717}
\showDOI{\tempurl}


\bibitem[\protect\citeauthoryear{Buterin, Reijsbergen, Leonardos, and
  Piliouras}{Buterin et~al\mbox{.}}{2019}]%
        {CasperIncentives}
\bibfield{author}{\bibinfo{person}{Vitalik Buterin}, \bibinfo{person}{Daniel
  Reijsbergen}, \bibinfo{person}{Stefanos Leonardos}, {and}
  \bibinfo{person}{Georgios Piliouras}.} \bibinfo{year}{2019}\natexlab{}.
\newblock \showarticletitle{Incentives in Ethereum’s Hybrid Casper Protocol}.
  In \bibinfo{booktitle}{\emph{2019 IEEE International Conference on Blockchain
  and Cryptocurrency (ICBC)}}. \bibinfo{publisher}{IEEE},
  \bibinfo{address}{Seoul, South Korea}, \bibinfo{pages}{236--244}.
\newblock
\urldef\tempurl%
\url{https://doi.org/10.1109/BLOC.2019.8751241}
\showDOI{\tempurl}


\bibitem[\protect\citeauthoryear{Daian, Goldfeder, Kell, Li, Zhao, Bentov,
  Breidenbach, and Juels}{Daian et~al\mbox{.}}{2019a}]%
        {FlashBoys2.0}
\bibfield{author}{\bibinfo{person}{Philip Daian}, \bibinfo{person}{Steven
  Goldfeder}, \bibinfo{person}{Tyler Kell}, \bibinfo{person}{Yunqi Li},
  \bibinfo{person}{Xueyuan Zhao}, \bibinfo{person}{Iddo Bentov},
  \bibinfo{person}{Lorenz Breidenbach}, {and} \bibinfo{person}{Ari Juels}.}
  \bibinfo{year}{2019}\natexlab{a}.
\newblock \bibinfo{title}{Flash Boys 2.0: Frontrunning, Transaction Reordering,
  and Consensus Instability in Decentralized Exchanges}.
\newblock \bibinfo{howpublished}{\url{https://arxiv.org/abs/1904.05234}}.
\newblock
\showeprint[arxiv]{1904.05234}
\newblock
\shownote{Retrieved: 19/05/2021.}


\bibitem[\protect\citeauthoryear{Daian, Pass, and Shi}{Daian
  et~al\mbox{.}}{2019b}]%
        {SnowWhite}
\bibfield{author}{\bibinfo{person}{Phil Daian}, \bibinfo{person}{Rafael Pass},
  {and} \bibinfo{person}{Elaine Shi}.} \bibinfo{year}{2019}\natexlab{b}.
\newblock \showarticletitle{Snow White: Robustly Reconfigurable Consensus and
  Applications to Provably Secure Proof of Stake}. In
  \bibinfo{booktitle}{\emph{Financial Cryptography and Data Security}}.
  \bibinfo{publisher}{Springer International Publishing},
  \bibinfo{address}{Cham}, \bibinfo{pages}{23--41}.
\newblock
\urldef\tempurl%
\url{https://doi.org/10.1007/978-3-030-32101-7_2}
\showDOI{\tempurl}


\bibitem[\protect\citeauthoryear{Dwork, Lynch, and Stockmeyer}{Dwork
  et~al\mbox{.}}{1988}]%
        {PartialSynchronyDwork}
\bibfield{author}{\bibinfo{person}{Cynthia Dwork}, \bibinfo{person}{Nancy
  Lynch}, {and} \bibinfo{person}{Larry Stockmeyer}.}
  \bibinfo{year}{1988}\natexlab{}.
\newblock \showarticletitle{Consensus in the Presence of Partial Synchrony}.
\newblock \bibinfo{journal}{\emph{J. ACM}} \bibinfo{volume}{35},
  \bibinfo{number}{2} (\bibinfo{date}{April} \bibinfo{year}{1988}),
  \bibinfo{pages}{288–323}.
\newblock
\showISSN{0004-5411}
\urldef\tempurl%
\url{https://doi.org/10.1145/42282.42283}
\showDOI{\tempurl}


\bibitem[\protect\citeauthoryear{Eyal and Sirer}{Eyal and Sirer}{2018}]%
        {SelfishMining}
\bibfield{author}{\bibinfo{person}{Ittay Eyal} {and}
  \bibinfo{person}{Emin~G\"{u}n Sirer}.} \bibinfo{year}{2018}\natexlab{}.
\newblock \showarticletitle{Majority is Not Enough: Bitcoin Mining is
  Vulnerable}.
\newblock \bibinfo{journal}{\emph{Commun. ACM}} \bibinfo{volume}{61},
  \bibinfo{number}{7} (\bibinfo{date}{June} \bibinfo{year}{2018}),
  \bibinfo{pages}{95–102}.
\newblock
\showISSN{0001-0782}
\urldef\tempurl%
\url{https://doi.org/10.1145/3212998}
\showDOI{\tempurl}


\bibitem[\protect\citeauthoryear{Fanti, Kogan, Oh, Ruan, Viswanath, and
  Wang}{Fanti et~al\mbox{.}}{2019}]%
        {PoSCompounding}
\bibfield{author}{\bibinfo{person}{Giulia Fanti}, \bibinfo{person}{Leonid
  Kogan}, \bibinfo{person}{Sewoong Oh}, \bibinfo{person}{Kathleen Ruan},
  \bibinfo{person}{Pramod Viswanath}, {and} \bibinfo{person}{Gerui Wang}.}
  \bibinfo{year}{2019}\natexlab{}.
\newblock \showarticletitle{Compounding of Wealth in Proof-of-Stake
  Cryptocurrencies}. In \bibinfo{booktitle}{\emph{Financial Cryptography and
  Data Security}}, \bibfield{editor}{\bibinfo{person}{Ian Goldberg} {and}
  \bibinfo{person}{Tyler Moore}} (Eds.). \bibinfo{publisher}{Springer
  International Publishing}, \bibinfo{address}{Cham}, \bibinfo{pages}{42--61}.
\newblock
\showISBNx{978-3-030-32101-7}


\bibitem[\protect\citeauthoryear{Kiayias, Russell, David, and
  Oliynykov}{Kiayias et~al\mbox{.}}{2017}]%
        {Ouroboros}
\bibfield{author}{\bibinfo{person}{Aggelos Kiayias}, \bibinfo{person}{Alexander
  Russell}, \bibinfo{person}{Bernardo David}, {and} \bibinfo{person}{Roman
  Oliynykov}.} \bibinfo{year}{2017}\natexlab{}.
\newblock \showarticletitle{Ouroboros: A Provably Secure Proof-of-Stake
  Blockchain Protocol}. In \bibinfo{booktitle}{\emph{Advances in Cryptology --
  CRYPTO 2017}}, \bibfield{editor}{\bibinfo{person}{Jonathan Katz} {and}
  \bibinfo{person}{Hovav Shacham}} (Eds.). \bibinfo{publisher}{Springer
  International Publishing}, \bibinfo{address}{Cham},
  \bibinfo{pages}{357--388}.
\newblock
\showISBNx{978-3-319-63688-7}


\bibitem[\protect\citeauthoryear{Kothapalli, Miller, and Borisov}{Kothapalli
  et~al\mbox{.}}{2017}]%
        {SmartCast}
\bibfield{author}{\bibinfo{person}{Abhiram Kothapalli}, \bibinfo{person}{Andrew
  Miller}, {and} \bibinfo{person}{Nikita Borisov}.}
  \bibinfo{year}{2017}\natexlab{}.
\newblock \showarticletitle{SmartCast: An incentive compatible consensus
  protocol using smart contracts}. In \bibinfo{booktitle}{\emph{Financial
  Cryptography and Data Security - FC 2017 International Workshops, Revised
  Selected Papers}}, \bibfield{editor}{\bibinfo{person}{Andrew Miller},
  \bibinfo{person}{Michael Brenner}, \bibinfo{person}{Kurt Rohloff},
  \bibinfo{person}{Joseph Bonneau}, \bibinfo{person}{Vanessa Teague},
  \bibinfo{person}{Andrea Bracciali}, \bibinfo{person}{Massimiliano Sala},
  \bibinfo{person}{Federico Pintore}, \bibinfo{person}{Markus Jakobsson}, {and}
  \bibinfo{person}{Peter~{Y.A.} Ryan}} (Eds.).
  \bibinfo{publisher}{Springer-Verlag Berlin Heidelberg},
  \bibinfo{address}{Sliema, Malta}, \bibinfo{pages}{536--552}.
\newblock
\showISBNx{9783319702773}
\urldef\tempurl%
\url{https://doi.org/10.1007/978-3-319-70278-0_34}
\showDOI{\tempurl}


\bibitem[\protect\citeauthoryear{Kwon}{Kwon}{2014}]%
        {Tendermint}
\bibfield{author}{\bibinfo{person}{Jae Kwon}.} \bibinfo{year}{2014}\natexlab{}.
\newblock \bibinfo{title}{Tendermint: Consensus without Mining}.
\newblock \bibinfo{howpublished}{\url{https://tendermint.com/static/docs}}.
\newblock
\newblock
\shownote{Retrieved: 19/05/2021.}


\bibitem[\protect\citeauthoryear{Lev-Ari, Spiegelman, Keidar, and
  Malkhi}{Lev-Ari et~al\mbox{.}}{2020}]%
        {Fairledger}
\bibfield{author}{\bibinfo{person}{Kfir Lev-Ari}, \bibinfo{person}{Alexander
  Spiegelman}, \bibinfo{person}{Idit Keidar}, {and} \bibinfo{person}{Dahlia
  Malkhi}.} \bibinfo{year}{2020}\natexlab{}.
\newblock \showarticletitle{{FairLedger: A Fair Blockchain Protocol for
  Financial Institutions}}. In \bibinfo{booktitle}{\emph{23rd International
  Conference on Principles of Distributed Systems (OPODIS 2019)}}
  \emph{(\bibinfo{series}{Leibniz International Proceedings in Informatics
  (LIPIcs)}, Vol.~\bibinfo{volume}{153})},
  \bibfield{editor}{\bibinfo{person}{Pascal Felber}, \bibinfo{person}{Roy
  Friedman}, \bibinfo{person}{Seth Gilbert}, {and} \bibinfo{person}{Avery
  Miller}} (Eds.). \bibinfo{publisher}{Schloss Dagstuhl--Leibniz-Zentrum fuer
  Informatik}, \bibinfo{address}{Dagstuhl, Germany},
  \bibinfo{pages}{4:1--4:17}.
\newblock
\showISBNx{978-3-95977-133-7}
\showISSN{1868-8969}
\urldef\tempurl%
\url{https://doi.org/10.4230/LIPIcs.OPODIS.2019.4}
\showDOI{\tempurl}


\bibitem[\protect\citeauthoryear{Liu, Luong, Wang, Niyato, Wang, Liang, and
  Kim}{Liu et~al\mbox{.}}{2019}]%
        {SOKGameTheoryBlockchain}
\bibfield{author}{\bibinfo{person}{Ziyao Liu}, \bibinfo{person}{Nguyen~Cong
  Luong}, \bibinfo{person}{Wenbo Wang}, \bibinfo{person}{Dusit Niyato},
  \bibinfo{person}{Ping Wang}, \bibinfo{person}{Ying-Chang Liang}, {and}
  \bibinfo{person}{Dong~In Kim}.} \bibinfo{year}{2019}\natexlab{}.
\newblock \showarticletitle{A Survey on Blockchain: A Game Theoretical
  Perspective}.
\newblock \bibinfo{journal}{\emph{IEEE Access}}  \bibinfo{volume}{7}
  (\bibinfo{year}{2019}), \bibinfo{pages}{47615--47643}.
\newblock


\bibitem[\protect\citeauthoryear{Moscibroda, Schmid, and
  Wattenhofer}{Moscibroda et~al\mbox{.}}{2006}]%
        {SelfishMeetsEvil}
\bibfield{author}{\bibinfo{person}{Thomas Moscibroda}, \bibinfo{person}{Stefan
  Schmid}, {and} \bibinfo{person}{Roger Wattenhofer}.}
  \bibinfo{year}{2006}\natexlab{}.
\newblock \showarticletitle{When Selfish Meets Evil: Byzantine Players in a
  Virus Inoculation Game}.
  \bibinfo{howpublished}{\url{https://doi.org/10.1145/1146381.1146391}}. In
  \bibinfo{booktitle}{\emph{Proceedings of the Twenty-Fifth Annual ACM
  Symposium on Principles of Distributed Computing}} (Denver, Colorado, USA)
  \emph{(\bibinfo{series}{PODC '06})}. \bibinfo{publisher}{Association for
  Computing Machinery}, \bibinfo{address}{New York, NY, USA},
  \bibinfo{pages}{35--44}.
\newblock
\showISBNx{1595933840}


\bibitem[\protect\citeauthoryear{Moscibroda, Schmid, and
  Wattenhofer}{Moscibroda et~al\mbox{.}}{2009}]%
        {PriceOfMalice}
\bibfield{author}{\bibinfo{person}{Thomas Moscibroda}, \bibinfo{person}{Stefan
  Schmid}, {and} \bibinfo{person}{Roger Wattenhofer}.}
  \bibinfo{year}{2009}\natexlab{}.
\newblock \showarticletitle{{The Price of Malice: A Game-Theoretic Framework
  for Malicious Behavior}}.
\newblock
  \bibinfo{howpublished}{\url{https://doi.org/10.1080/15427951.2009.10129181}}.
\newblock \bibinfo{journal}{\emph{Internet Mathematics}} \bibinfo{volume}{6},
  \bibinfo{number}{2} (\bibinfo{year}{2009}), \bibinfo{pages}{125--156}.
\newblock
\urldef\tempurl%
\url{https://doi.org/10.1080/15427951.2009.10129181}
\showDOI{\tempurl}


\bibitem[\protect\citeauthoryear{Nakamoto}{Nakamoto}{2008}]%
        {Bitcoin}
\bibfield{author}{\bibinfo{person}{Satoshi Nakamoto}.}
  \bibinfo{year}{2008}\natexlab{}.
\newblock \bibinfo{title}{Bitcoin: A Peer-to-Peer Electronic Cash System}.
\newblock \bibinfo{howpublished}{\url{https://bitcoin.org/bitcoin.pdf}}.
\newblock
\newblock
\shownote{Retrieved: 19/05/2021.}


\bibitem[\protect\citeauthoryear{Negy, Rizun, and Sirer}{Negy
  et~al\mbox{.}}{2020}]%
        {IntermittentSelfishMinig}
\bibfield{author}{\bibinfo{person}{Kevin~Alarc{\'o}n Negy},
  \bibinfo{person}{Peter~R. Rizun}, {and} \bibinfo{person}{Emin~G{\"u}n
  Sirer}.} \bibinfo{year}{2020}\natexlab{}.
\newblock \showarticletitle{Selfish Mining Re-Examined}. In
  \bibinfo{booktitle}{\emph{Financial Cryptography and Data Security}},
  \bibfield{editor}{\bibinfo{person}{Joseph Bonneau} {and}
  \bibinfo{person}{Nadia Heninger}} (Eds.). \bibinfo{publisher}{Springer
  International Publishing}, \bibinfo{address}{Cham}, \bibinfo{pages}{61--78}.
\newblock
\showISBNx{978-3-030-51280-4}


\bibitem[\protect\citeauthoryear{Nisan, Roughgarden, Tardos, and
  Vazirani}{Nisan et~al\mbox{.}}{2007}]%
        {AlgoGameTheory}
\bibfield{author}{\bibinfo{person}{Noam Nisan}, \bibinfo{person}{Tim
  Roughgarden}, \bibinfo{person}{Eva Tardos}, {and} \bibinfo{person}{Vijay~V.
  Vazirani}.} \bibinfo{year}{2007}\natexlab{}.
\newblock \bibinfo{booktitle}{\emph{Algorithmic Game Theory}}.
\newblock \bibinfo{publisher}{Cambridge University Press},
  \bibinfo{address}{Cambridge}.
\newblock
\showISBNx{0521872820}


\bibitem[\protect\citeauthoryear{Pass and Shi}{Pass and Shi}{2017}]%
        {FruitChains}
\bibfield{author}{\bibinfo{person}{Rafael Pass} {and} \bibinfo{person}{Elaine
  Shi}.} \bibinfo{year}{2017}\natexlab{}.
\newblock \showarticletitle{FruitChains: A Fair Blockchain}. In
  \bibinfo{booktitle}{\emph{Proceedings of the ACM Symposium on Principles of
  Distributed Computing}} (Washington, DC, USA) \emph{(\bibinfo{series}{PODC
  ’17})}. \bibinfo{publisher}{Association for Computing Machinery},
  \bibinfo{address}{New York, NY, USA}, \bibinfo{pages}{315--324}.
\newblock
\showISBNx{9781450349925}
\urldef\tempurl%
\url{https://doi.org/10.1145/3087801.3087809}
\showDOI{\tempurl}


\bibitem[\protect\citeauthoryear{Rosu and Saleh}{Rosu and Saleh}{2020}]%
        {EvolutionOfSharesPoS}
\bibfield{author}{\bibinfo{person}{Ioanid Rosu} {and} \bibinfo{person}{Fahad
  Saleh}.} \bibinfo{year}{2020}\natexlab{}.
\newblock \bibinfo{title}{Evolution of Shares in a Proof-of-Stake
  Cryptocurrency}.
\newblock \bibinfo{howpublished}{\url{http://dx.doi.org/10.2139/ssrn.3377136}}.
\newblock
\newblock
\shownote{Retrieved: 19/05/2021.}


\bibitem[\protect\citeauthoryear{Roughgarden}{Roughgarden}{2020}]%
        {EIP1559roughgarden}
\bibfield{author}{\bibinfo{person}{Tim Roughgarden}.}
  \bibinfo{year}{2020}\natexlab{}.
\newblock \bibinfo{title}{Transaction Fee Mechanism Design for the Ethereum
  Blockchain: An Economic Analysis of EIP-1559}.
\newblock \bibinfo{howpublished}{\url{https://arxiv.org/pdf/2012.00854 }}.
\newblock
\showeprint[arxiv]{2012.00854}
\newblock
\shownote{Retrieved: 18/05/2021.}


\bibitem[\protect\citeauthoryear{Saleh}{Saleh}{2020}]%
        {BlockchainWithoutWastePoS}
\bibfield{author}{\bibinfo{person}{Fahad Saleh}.}
  \bibinfo{year}{2020}\natexlab{}.
\newblock \showarticletitle{Blockchain Without Waste: Proof-of-Stake}.
  \bibinfo{howpublished}{\url{http://dx.doi.org/10.2139/ssrn.3183935}}. In
  \bibinfo{booktitle}{\emph{Review of Financial Studies}},
  Vol.~\bibinfo{volume}{34}. \bibinfo{publisher}{March},
  \bibinfo{address}{2021}, \bibinfo{pages}{1156--1190}.
\newblock


\bibitem[\protect\citeauthoryear{Sliwinski and Wattenhofer}{Sliwinski and
  Wattenhofer}{2020}]%
        {BlockchainsCannotRelyonHonesty}
\bibfield{author}{\bibinfo{person}{Jakub Sliwinski} {and}
  \bibinfo{person}{Roger Wattenhofer}.} \bibinfo{year}{2020}\natexlab{}.
\newblock \bibinfo{title}{Blockchains Cannot Rely on Honesty}.
\newblock
  \bibinfo{howpublished}{\url{https://disco.ethz.ch/courses/fs19/sirocco/honesty.pdf}}.
\newblock
\newblock
\shownote{Retrieved: 21/05/2021.}


\end{thebibliography}
